%% file: main.tex
\documentclass[11pt]{article}

\usepackage{amsmath,amssymb,amsfonts,amsthm,epsfig}
\usepackage[usenames,dvipsnames]{xcolor}
\usepackage{bm,xspace}
\usepackage{tcolorbox}
\usepackage{cancel}
\usepackage{fullpage}
\usepackage{framed}
\usepackage{liyang}
\usepackage{verbatim}
\usepackage{enumitem}
\usepackage{array}
\usepackage{multirow}
\usepackage{afterpage}
\usepackage{mathrsfs}
\usepackage{pifont} 
\usepackage{chngpage}
\usepackage[normalem]{ulem}
\usepackage{boxedminipage}
\usepackage{caption}
\usepackage{subcaption}

\usepackage{forest}

\usepackage{algorithm}
\usepackage{algorithmicx}
\usepackage{algpseudocode}

\usepackage{tikz}

\usetikzlibrary{calc,through,backgrounds,decorations.pathreplacing, calligraphy,arrows.meta}
\usetikzlibrary{positioning,chains,fit,shapes}

\usepackage{tikz-cd}

\usepackage{pgfplots}
\pgfplotsset{width=8cm,compat=newest}

\def\colorful{1}

\ifnum\colorful=1

\fi
\ifnum\colorful=0

\fi

\crefname{fact}{fact}{facts}

\newcommand{\error}{\mathrm{error}}

\newcommand{\adv}{\mathrm{Adv}}
\newcommand{\Adv}{\mathrm{Adv}}

\newcommand{\mcD}{\mathcal D}

\newcommand{\mcS}{\mathcal S}

\def\maj{\textsc{Maj}}

\def\XOR{\mathrm{XOR}}

\newcommand{\Stab}{\mathrm{Stab}}
\newcommand{\stab}{\mathrm{Stab}}
\newcommand{\unstab}{\mathrm{UnbalStab}}

\newcommand{\vrho}{\vec{\rho}}
\newcommand{\rhosim}{\overset{\vrho}{\sim}}
\newcommand{\absim}{\overset{a,b}{\sim}}

\newcommand{\smalle}{\eps_{\mathrm{small}}}
\newcommand{\largee}{\eps_{\mathrm{large}}}
\newcommand{\yese}{\eps_{\mathrm{yes}}}
\newcommand{\noe}{\eps_{\mathrm{no}}}
\newcommand{\weak}{\mathcal{T}_{\mathrm{weak}}}
\newcommand{\strong}{\mathcal{T}_{\mathrm{strong}}}

\newcommand{\rhoup}{\rho^{(\mathrm{AM})}}
\newcommand{\rholow}{\rho^{(\mathrm{GM})}}

\DeclarePairedDelimiter\ceil{\lceil}{\rceil}

\DeclarePairedDelimiter\abs{|}{|}

\DeclarePairedDelimiter\paren{(}{)}
\DeclarePairedDelimiter\bracket{[}{]}

\newlist{enumprop}{enumerate}{1} 
\setlist[enumprop]{label=\arabic*.,ref=\theproposition.\arabic*}

\makeatletter
\newtheorem*{rep@theorem}{\rep@title}
\newcommand{\newreptheorem}[2]{
\newenvironment{rep#1}[1]{
 \def\rep@title{#2 \ref{##1}}
 \begin{rep@theorem}\itshape}
 {\end{rep@theorem}}}
\makeatother

\newreptheorem{theorem}{Theorem}

\newcommand{\pparagraph}[1]{\bigskip \noindent {\bf {#1}}}

\title{A strong composition theorem for junta complexity \\ and the boosting of property testers \vspace{10pt} }

\author{ 
Guy Blanc \vspace{6pt} \\ 
\hspace{-7pt} {\sl Stanford} \and 
Caleb Koch \vspace{6pt} \\ 
\hspace{-10pt} { {\sl Stanford}} \and 
 Carmen Strassle \vspace{6pt} \\
 \hspace{-5pt} { {\sl Stanford}} \vspace{15pt}
 \and 
Li-Yang Tan \vspace{6pt}  \\
\hspace{-10pt} {{\sl Stanford}}
}

\date{\small{\today}}

\begin{document}

\maketitle

\input{Intro}

\input{Preliminaries}

\input{StrongDirectProduct}

\input{Symmetric}

\input{property_testing}

 \section*{Acknowledgments}

 We thank the FOCS reviewers for their helpful comments and feedback. The authors are supported by NSF awards 1942123, 2211237, 2224246 and a Google Research Scholar award. Caleb is also supported by an NDSEG fellowship, and Carmen by a Stanford Computer Science Distinguished Fellowship.

\bibliographystyle{alpha}
\bibliography{ref}

\appendix
\input{Counterexamples}

\end{document}

%% file: Intro.tex
\begin{abstract} 
We prove a {\sl strong composition theorem} for junta complexity and show how such theorems can be used to generically boost the performance of property testers. 

The $\varepsilon$-approximate junta complexity of a function $f$ is the smallest integer $r$ such that $f$ is $\varepsilon$-close to a function that depends only on $r$ variables. A strong composition theorem states that if $f$ has large $\varepsilon$-approximate junta complexity, then $g \circ f$ has {\sl even larger} $\varepsilon’$-approximate junta complexity, {\sl even for} $\varepsilon’ \gg \varepsilon$. We develop a fairly complete understanding of this behavior, proving that the junta complexity of $g \circ f$ is characterized by that of $f$ along with the multivariate noise sensitivity of $g$. For the important case of symmetric functions $g$, we relate their multivariate noise sensitivity to the simpler and well-studied case of univariate noise sensitivity. 

We then show how strong composition theorems yield {\sl boosting algorithms} for property testers: with a strong composition theorem for any class of functions, a large-distance tester for that class is immediately upgraded into one for small distances. Combining our contributions yields a booster for junta testers, and with it new implications for junta testing. This is the first boosting-type result in property testing, and we hope that the connection to composition theorems adds compelling motivation to the study of both topics.
\end{abstract} 

\thispagestyle{empty}

\newpage 

\tableofcontents
\thispagestyle{empty}

\newpage 

\setcounter{page}{1}

\section{Introduction}

The growth in the sizes of modern datasets is both a blessing and a curse. These datasets, many of which now come with billions of features, contain a wealth of information that machine learning algorithms seek to tap into. On the other hand, their size stands in the way of the opportunities they present, as many of the algorithms that we would like to run on them simply cannot handle their dimensionality. 

Thankfully, for many tasks of interest the vast majority of features are irrelevant. This motivates the design of algorithms that are able to quickly home in on the small number of relevant features, and whose efficiency scales gracefully with the number of such features. Already in the early 1990s Blum~\cite{Blu94}  (see also~\cite{BL97,Blu03}) proposed the clean theoretical challenge  of learning an unknown {\sl $r$-junta}, a function that depends on $r\ll n$ many of its $n$ variables. Quoting~\cite{Blu94}, “It is my belief that some of the most central open problems in computational learning theory are, at their core, questions about finding relevant variables.” This is now known simply as {\sl the junta problem} and is the subject of intensive study~\cite{BHL95,MOS04,AR07,AS07,KLMMV09,AM10,ST11,Bel15,Val15,ABR16,CJLW21,CNY23}, having distinguished itself as ``the single most important open question in uniform distribution learning"~\cite{MOS04}. 

The premise of the junta problem suggests an even more basic algorithmic problem, that of determining if an unknown function is even an $r$-junta to begin with. This is the problem of {\sl testing} juntas, introduced by Fischer, Kindler, Ron, Safra, and Samorodnitsky~\cite{FKRSS04} and subsequently studied in numerous works~\cite{CG04,AS07,Bla08,Bla09,CGSM11,BGSMdW13,STW15,ABR16,ABRdW16,BKT18,Sag18,LCSSYX18,CSTWX18,BCELR18,Bsh19,Bel19,LW19,DMN19,CJLW21,ITW21,PRW22,CNY23}.  Junta testers are also at the heart of the best known testers for numerous {\sl other} classes of functions, the key insight being that many functions are well-approximated by small juntas (see~\cite{DLMORSW07,Ser10} and Chapter 5 of~\cite{Ron10} for more on this connection). The surveys by Blais~\cite{Bla10,Bla16} give broad overviews of various junta testers and their applications throughout theoretical computer science.

\paragraph{This work.}  These algorithmic applications motivate the study of approximability by small juntas as a complexity measure. For a function $f : \bits^n \to \bits$ and a distribution $\mathcal{D}$ over $\bits^n$, the $\varepsilon$-approximate junta complexity of $f$ with respect to $\mathcal{D}$, denoted $J_{\mathcal{D}}(f,\varepsilon)$, is the smallest integer $r$ such that $f$ is $\varepsilon$-close to an $r$-junta. Among the most basic questions one can ask about any complexity measure of functions is how it behaves under {\sl composition}. In the first part of this paper we develop, from the ground up, a fairly complete understanding of this question for junta complexity. We prove a near-optimal composition theorem (\Cref{thm:composition theorem intro 1})  that is built on notions of {\sl noise stability}, both classical and new. In the second part we draw a general connection (\Cref{thm:generic boost intro}) between the type of composition theorem that we prove---a {\sl strong} composition theorem, which we will soon define---and property testing, showing how they can be used to design the first generic boosters for property testers. Combining our two main contributions yields new implications for junta testing. 

\section{Our results and techniques} 

\subsection{First main result: A strong composition theorem for junta complexity}

Composition theorems are statements about {\sl hardness amplification}: the goal is to understand the extent to which the disjoint composition $(g \circ f)(x) \coloneqq g(f(x^{(1)}),\ldots,f(x^{(k)}))$ is more complex than $f$ itself, and how this depends on intrinsic properties of the combining function $g$. For approximate measures such has junta complexity, we are furthermore interested in {\sl strong} composition theorems, statements of the form: 
\begin{equation} 
\label{eq:strong amplification intro} 
J_{\mathcal{D}^k}(g\circ f, \varepsilon_{\mathrm{large}})\gg J_{\mathcal{D}}(f, \varepsilon_{\mathrm{small}})\ \ \text{even for}\ \  \varepsilon_{\mathrm{large}} \gg \varepsilon_{\mathrm{small}}. \tag{$\diamondsuit$}
\end{equation} 
In words, the composed function requires much more resources---in our case, much larger junta approximators---even if one only seeks a much coarser approximation. Strong composition theorems stand in contrast to weak ones that only amplify hardness with respect to one of the two parameters, either resources or approximation quality only. The canonical example in this context is Yao’s XOR lemma~\cite{Yao82}, which says that if $f$ is mildly hard to approximate with size-$s$ circuits, then $\mathrm{XOR} \circ f$ is extremely hard to approximate with size-$s’$ circuits. A long-recognized downside of this important result, inherent to all known proofs of it~\cite{Lev85,GNW11,Imp95,IW97} and its generalizations to arbitrary combining functions~\cite{OD02}, is the fact that it is only known to hold for $s’ \ll s$, whereas intuitively it should hold even for $s’ \gg s$. 

Composition theorems, both weak and strong, have been studied for a variety of complexity measures
 but appear to have been underexplored for junta complexity. One reason may be that the question appears deceptively simple. Indeed, things {\sl are} completely straightforward in the zero-error setting, where we have the intuitive identity $J(g \circ f, 0) = J(g,0)\cdot J(f,0)$. However, we show that the question becomes surprisingly intricate once error is allowed.

\subsubsection{Context and motivation: Counterexamples to natural composition theorems}
\label{subsec:counterexamples}

The question proves to be tricky even in the special case where the combining function $g$ is symmetric. We now state a sequence of three seemingly intuitive conjectures for this special case. While false, these conjectures and their counterexamples will motivate and lead us to the statement of our actual composition theorem. (Details and proofs of the counterexamples discussed in this section are given in~\Cref{app:counterexamples}.)

The following notation will be useful for us throughout this paper: 

\paragraph{Notation.} For a function $f : \bits^n\to\bits$, distribution $\mathcal{D}$ over $\bits^n$, and integer $r$, we write $\tilde{f}_{\mathcal{D},r}$ to denote the best $r$-junta approximator of $f$ with respect to $\mathcal{D}$. When $\mathcal{D}$ is clear from context, we simply write $\tilde{f}_r$.   

\paragraph{Conjecture 1.}  It will be convenient for us to consider composition theorems in their contrapositive form. Suppose we would like to approximate $g \circ f$ with an $R$-junta, say with respect to the uniform distribution. If $g$ is a $k$-variable symmetric function, how would we go about constructing an approximator that achieves the highest accuracy possible? Since $g$ is symmetric, one may be inclined to divide the “junta budget” of $R$ evenly among the $k$ inner functions and conjecture that
\[ 
g \circ \tilde{f}_{R/k} = g(\tilde{f}_{R/k},\ldots,\tilde{f}_{R/k})
\] 
achieves the best, or close to the best, accuracy among all $R$-junta approximators. 

However, this is badly false. Let $g$ be the $k$-variable Majority function and $f$ the $n$-variable Parity function. For any choice of $R$ satisfying $R/k < n$ (i.e.~each inner Parity receiving a budget that falls short of its arity), we have $\mathrm{Pr}[g\circ\tilde{f}_{R/k}\ne g\circ f] = \frac1{2}.$ This is because it is “all or nothing” when it comes to approximating Parity: no ($n-1$)-junta can achieve accuracy better than that of a constant approximator. The best strategy is therefore to allocate a full budget of $n$ to as many of the inner Parities as possible (i.e.~$R/n$ many of them), and a budget of  zero to the others. This shows a gap of $\frac1{2}$ versus $1-o(1)$ in the accuracies of the “divide budget equally” strategy and the optimal one. 

\paragraph{Conjecture 2.}  In light of this counterexample, one may then conjecture that the best strategy is to partition the junta budget optimally among the $k$ inner functions and feed the respective approximators of $f$ into $g$. That is, the conjecture is that the best approximator is of the form:  
\[ 
g(\tilde{f}_{r_1},\ldots,\tilde{f}_{r_k}) \text{ where } \sum_{i=1}^k r_i = R. 
\] 
While this is true for our example above, it is again badly false in general. In fact, the error of such an approximator can be close to $1$, even worse than the trivial bound of $\le \frac1{2}$ achievable with a constant approximator. 

Our counterexample reveals another counterintuitive aspect of the overall problem. Consider an approximator for $g\circ f$ of the form $g(\tilde{f}_{r_1},\ldots,\tilde{f}_{r_k}).$ We show its approximation accuracy can {\sl increase} if we replace one of the inner approximators for $f$ with a {\sl worse} one: e.g.~if we replace $\tilde{f}_{r_1}$ with $\tilde{f}_{{r_1}’}$ where ${r_1}’ < r_1$. In more technical terms that we will soon define: while the noise stability of a function is, as one would expect, monotone in the noise rate, we show that the natural generalization of it where the corruption probabilities of $0$’s and $1$’s are decoupled (defined in~\Cref{sec:proof of lower bound on advantage}) is {\sl not} monotone.   

\paragraph{Conjecture 3.} Finally, we consider a conjecture that is far laxer than either of the previous ones. It simply states that the optimal approximator for the composed function $g\circ f$ is one of {\sl composed form}: 
\[ 
h(q^{(1)},\ldots,q^{(k)}) \text{ for some $h : \bits^k \to \bits$ and $q^{(1)},\ldots,q^{(k)} : \bits^n \to \bits$,} 
\] 
where the relevant variables of $q^{(i)}$ fall within the $i$th block of variables. 
We show (to our own surprise) that this conjecture is {\sl still} false: there are composed functions for which the optimal approximator is not of composed form.   However, unlike the first two conjectures, our work shows that this conjecture is morally true in a precise sense. 

\subsubsection{Our Strong Composition Theorem} 

Our strong composition theorem implies a close quantitative relationship between the error of the optimal approximator and that of the optimal  composed form approximator, and indeed one with a specific structure that we call {\sl canonical}: 

\begin{definition}[Canonical composed form approximators] 
We say that a composed form approximator for $g\circ f$ is {\sl canonical} if it is of the form:
\[ h(\tilde{f}_{r_1},\ldots,\tilde{f}_{r_k}), \] 
where $h : \bits^k\to\bits$ is the function:  
\[ h(y) = \sign\bigg(\Ex_{\bx\sim \mathcal{D}^k}\Big[ (g\circ f)(\bx)\mid \text{$y_i = \tilde{f}_{r_i}(\bx^{(i)})$ for all $i\in [k]$} \Big]\bigg).   \]
\end{definition} 

For intuition regarding the choice of $h$, we note that for the fixed $k$-tuple of functions $\tilde{f}_{r_1},\ldots,\tilde{f}_{r_k}$, it is the combining function that minimizes error with respect to $g\circ f$. 
Canonical composed form approximators are therefore ones whose individual components are ``locally" optimal: each $\tilde{f}_{r_i}$ is the optimal $r_i$-junta approximator for $f$, and $h$ the optimal way of combining the $f_{r_i}$'s. Our strong composition theorem will say that we can get very close to the {\sl globally} optimal approximator this way.

The notion of {\sl noise stability} is central to our work: 
\begin{definition}[Multivariate noise stability]
    \label{def:multi-noise}
    For any $\mu \in (-1,1)$ and vector $\vec{\rho} \in [0,1]^k$, we define the multivariate noise stability of $g$ as
    \begin{equation*}
        \Stab_{\mu,\vec{\rho}}(g) = \Ex[g(\by)g(\bz)]
    \end{equation*}
    where independently for each $i \in [k]$, we  draw $(\by_i$, $\bz_i)$ as follows: Using $\pi_{\mu}$ to denote the unique distribution supported on $\bits$ with mean $\mu$, $\by_i \sim \pi_\mu$, and
    \begin{equation*}
        \bz_i = \begin{cases}
            \by_i &\text{w.p.~}\vec{\rho}_i\\
            \text{Independent draw from $\pi_{\mu}$} &\text{w.p.~}1 - \vec{\rho}_i.
        \end{cases}
    \end{equation*}
When $\mu = 0$ we simply write $\Stab_{\vec{\rho}}(g)$. \end{definition}

This definition allows for a different noise rate for each coordinate, generalizing the more commonly studied definition where the noise rates are the same for every coordinate (see e.g.~Chapter 2 of~\cite{ODBook}). We use the terms {\sl multivariate} noise stability and {\sl univariate} noise stability to distinguish these definitions. Even in the case of symmetric combining functions $g$, our strong composition theorem will naturally involve its multivariate noise stability (necessarily so, as already suggested by the counterexample to Conjecture 1). 

We present our strong composition theorem as a sequence of two parts that each carries a standalone message, the first of which formalizes the fact that the optimal canonical composed form approximator is a good proxy for the actual optimal approximator. It will be more convenient for us to state our results in terms of advantage instead of error, the two quantities being related via the identity $\text{advantage} = 1-2\cdot\text{error}$. Also, for notational clarity we only state here the special case where $f$ is balanced (i.e.~$\E_{\mathcal{D}}[f] = 0$). 
\medskip 

\begin{tcolorbox}[colback = white,arc=1mm, boxrule=0.25mm]
\begin{theorem}[Part I: Canonical composed form approximators are near optimal]
\label{thm:composition theorem intro 1} 
Let $f : \bits^n\to\bits$ and $g:\bits^k \to \bits$ be arbitrary functions and $\mathcal{D}$ be any distribution over $\bits^n$. Assume that $\E_{\mathcal{D}}[f]=0$. For the task of approximating $g \circ f$ under $\mathcal{D}^k$ with an $R$-junta, there is a correlation vector $\vec{\rho} \in [0,1]^k$ such that 
\begin{align*}
    \Stab_{\vec{\rho}}(g)^2 &\le  \textnormal{Advantage of optimal canonical composed form approximator} \\
    &\le \textnormal{Advantage of optimal approximator} \le \sqrt{\Stab_{\vec{\rho}}(g)}. 
\end{align*} 
\end{theorem} 
\end{tcolorbox} 
\medskip

For most applications of composition theorems, including those in this paper, the parameters of interest are such that the quartic gap between the upper and lower bounds above are inconsequential. (In particular, if the advantage of the optimal canonical composed form approximator diminishes to $0$ as $k$ grows,  our bounds imply that the same is true for the actual optimal approximator. Indeed, the two rates of convergence are the same up to a polynomial factor.) 

Part II of~\Cref{thm:composition theorem intro 1} elaborates on the correlation vector $\vec{\rho}$, showing how it is is determined by the junta complexity of $f$ and the noise stability of $g$: 

\medskip 


\begin{tcolorbox}[colback = white,arc=1mm, boxrule=0.25mm]
{\bf Theorem 1} (Part II: Explicit description of $\vec{\rho}$){\bf .} {\it The correlation vector $\vec{\rho} \in [0,1]^k$ in Part I is the vector that maximizes $\Stab_{\vec{\rho}}(g)$, subject to the constraint:
\[ \vec{\rho}_i = \Ex_{\mathcal{D}}\big[f\cdot \tilde{f}_{r_i}\big] \text{ for all $i\in [k]$ where } \ds \sum_{i=1}^k r_i = R. \]}
\end{tcolorbox} 
\medskip

Taken together, the two parts of~\Cref{thm:composition theorem intro 1} show that the junta complexity of $g\circ f$ is tightly characterized by the junta complexity of $f$ and the multivariate noise stability of $g$. It furthermore gives a simple and explicit strategy for constructing a near-optimal approximator: first partition the junta budget optimally among the $k$ inner functions; next approximate each inner function optimally with its allocated budget; and finally combine these approximators in the optimal way. 

Naturally, it would be preferable to understand the strategy for constructing the actual optimal approximator, but our counterexamples suggest that it defies a clean and interpretable description even for symmetric $g$ (indeed, even for $g$ being the And function).

\paragraph{Corollary: Highly noise sensitive functions strongly amplify junta complexity.}  \Cref{thm:composition theorem intro 1} yields a hardness amplification statement of the form~\ref{eq:strong amplification intro} the following way.  Suppose $f$ is mildly hard for $r$-juntas, i.e.~$\Pr[\tilde{f}_r \ne f] \ge \varepsilon_{\mathrm{small}}$. Our goal is to show that $g \circ f$ is extremely hard for $R$-juntas, $\Pr[(g\circ f)_R \ne g\circ f] \eqqcolon \eps_{\mathrm{large}} \gg \varepsilon_{\mathrm{small}}$, even for $R \gg r$.  For any partition of $R = \sum_{i=1}^k r_i$, at most a $0.999$-fraction of the $r_i$'s exceed $1.01R/ k\eqqcolon r$.  \Cref{thm:composition theorem intro 1} therefore tells us that the advantage of the optimal $R$-junta is upper bounded by 
\[ \sqrt{\Stab_{\vec{\rho}}(g)} \text{ where at least a $0.001$-fraction of $\vec{\rho}$'s coordinates are at most $1-2\cdot \eps_{\mathrm{small}}$.}  \]
(Equivalently, at least a $0.001$-fraction of coordinates receive at least an $\eps_{\mathrm{small}}$ amount of noise.) 

This motivates the following definition: 
\begin{definition}[Stability under partial noise] 
\label{def:delta-eps-NS}
The $(\delta,\eps)$-noise stability of a function $g:\bits^k\to\bits$ is the quantity
\[ \max\big\{ \Stab_{\vec{\rho}}(g)\colon  \text{at least a $\delta$-fraction of $\vec{\rho}$'s coordinates are at most $1-2\eps$}\big\}.\]   
By the monotonicity of noise stability, this maximum is achieved by a $\vec{\rho}$ with exactly a $\delta$-fraction of coordinates being exactly $1-2\eps$, and the remaining $(1-\delta)$-fraction being $1$.
\end{definition}

  We have sketched the following corollary of~\Cref{thm:composition theorem intro 1}: 

\begin{corollary}[Highly noise sensitive functions strongly amplify junta complexity]
\label{cor:intro amplify junta complexity}
Let $g : \bits^k \to \bits$ be a function whose $(\frac1{2},\eps_{\mathrm{small}})$-noise stability is at most $\tau$. Then for all functions $f$,
 \[ 
\mathop{\underbrace{J_{\mathcal{D}^k}\big(g\circ f, \tfrac1{2}(1-\sqrt{\tau})\big)}_{\textup{Junta complexity of $g\circ f$}}}_{\textup{for large error}} \ge 0.99k\cdot  \hspace{-12pt} \mathop{\underbrace{J_{\mathcal{D}}(f,\eps_{\mathrm{small}})}_{\textup{Junta complexity of $f$}}}_{\textup{for small error}}\hspace{-13pt}. 
\] 
\end{corollary} 
 In words, $g \circ f$ requires much larger junta approximators, an $\Omega(k)$ multiplicative factor more, even if we allow much larger error, $\frac1{2}(1- \sqrt{\tau}) \eqqcolon \eps_{\mathrm{large}}$ instead of $\eps_{\mathrm{small}}$. As two extreme examples of combining functions $g$, 
\begin{itemize}
    \item[$\circ$] The $(0.001,\eps_{\mathrm{small}})$-noise stability of the $k$-variable Parity function is $(1-2\cdot \eps_{\mathrm{small}})^{\Omega(k)}$, making it an excellent amplifier of junta complexity.
\item[$\circ$] The $(0.001,\eps_{\mathrm{small}})$-noise stability of a dictator function $g(x) = x_i$ is $1$, making it a terrible amplifier of junta complexity as one would expect: if $g$ is a dictator function then $g\circ f \equiv f$ is of course no more complex than $f$ itself.
\end{itemize} 

The partial-noise stability of these two specific examples are straightforward to compute, but the calculations quickly become unwieldy even for other basic functions. In addition to being a quantity of independent technical interest, the upcoming connections between strong composition theorems and the boosting of property testers will also motivate understanding the partial-noise stability of broad classes of functions beyond just parity and dictator. (Roughly speaking, to boost testers for a property $\mathcal{P}$ we need to analyze a function $g$ such that $\mathcal{P}$ is closed under $g$.)

Our next result is a general technique that yields sharp bounds on the partial-noise stability, and more generally the multivariate noise stability, of all {\sl symmetric} functions. 

\paragraph{The multivariate noise sensitivity of symmetric functions.} For a symmetric function $g : \bits^k \to\bits$ one intuits that its multivariate noise stability at a vector $\vec{\rho}\in [0,1]^k$ should be related to its univariate noise stability at a value $\rho^\star \in [0,1]$ that is an ``average" of the coordinates of $\vec{\rho}$. (This is certainly not true for general functions; consider for example the dictator function.)  Using techniques from the study of negative association, we formalize this intuition and prove that indeed it is sandwiched by the arithmetic and geometric means of the coordinates of $\vec{\rho}$: 

\begin{lemma}[Multivariate and univariate noise stabilities of symmetric functions]
\label{lem:intro multivariate to univariate}
Let $g : \bits^k\to\bits$ be a symmetric function, $\mu \in (-1,1)$, and $\vec{\rho}\in [0,1]^k$. Define
\[         \rholow \coloneqq \Bigg(\prod_{i \in [k]}\vec{\rho}_i\Bigg)^{1/k} \quad\text{and}\quad\ \ \rhoup \coloneqq  \frac{1}{k} \sum_{i \in [k]} \vec{\rho}_i.
   \] 
Then 
\[ \Stab_{\mu,\rholow}(g) \le  \Stab_{\mu,\vec{\rho}}(g) \le \Stab_{\mu,\rhoup}(g). \]
Furthermore, the lower bound holds under the weaker assumption that $g$ is transitive. 
\end{lemma} 

The more ``reasonable" $\vec{\rho}$ is, the closer the upper and lower bounds of~\Cref{lem:intro multivariate to univariate} are. In particular, we get the following bound on the $(\delta,\eps)$-noise stability of symmetric functions: 

\begin{corollary}[The $(\delta,\eps)$-noise stability of symmetric functions; informal]
\label{cor:delta-eps-noise-stability-intro}
    For any symmetric function $g:\bits^k \to \bits$, $\delta \in (0,1)$, and $\eps \in (0,1/2)$, the $(\delta, \eps)$-noise stability of $g$ is equal to $\stab_{\mu, \rho^\star}(g)$ for some $\rho^\star \in [0,1]$ satisfying
    \begin{equation*}
        1 - 2\delta\eps - O(\eps^2) \leq \rho^\star \leq 1 - 2\delta\eps.
    \end{equation*}
\end{corollary}

Recall that $\eps$ corresponds to the initial inapproximability factor $\eps_{\mathrm{small}}$ in~\Cref{cor:intro amplify junta complexity}, and so the additive gap of $O(\eps^2)$ between the upper and lower bounds is indeed small for our intended application.


\subsection{Second main result: Composition theorems and boosting of property testers} 

Composition theorems are most naturally thought of as statements about hardness amplification, and indeed that is how they are most commonly used. As our second main contribution, we show how they can be used fruitfully in their contrapositive form as meta-algorithms. In more detail, we show how they can be used to generically boost the performance guarantees of {\sl property testers}. While boosting is a story of success in both the theory and practice of machine learning, to our knowledge the analogous concept in property testing has not yet been considered. The connection that we draw can be instantiated with either strong or weak composition theorems, but as we now see, the parameters are qualitatively better in case of strong composition theorems.  

Within property testing, a major strand of research, initiated by Parnas, Ron, and Samorodnitsky~\cite{PRS02}, concerns testing  whether an unknown function has a {\sl concise representation}. Consider any parameterized property $\mathcal{P} = \{ \mathcal{P}_s\}_{s \in \mathbb{N}}$ of boolean functions: size-$s$ parities, size-$s$ juntas, size-$s$ decision trees, $s$-sparse polynomials over various fields, and so on. The task is as follows: 

\begin{definition}[$(\varepsilon,s,s')$-testing of $\mathcal{P}$ under $\mathcal{D}$]
\label{def:testing of P under D}
Given queries to an unknown function $f : \bits^n \to \bits$, access to i.i.d.~draws from a distribution $\mathcal{D}$, and parameters $s,s'\in \N$ and $\eps > 0$, distinguish between:
\begin{itemize} 
 \item[$\circ$] Yes: $f \in \mathcal{P}_s$
 \item[$\circ$] No: $f$ is $\varepsilon$-far under $\mathcal{D}$ from every function in $\mathcal{P}_{s'}$.
 \end{itemize} 
\end{definition} 
Note that the task is more challenging as $\varepsilon$ gets smaller, and as the gap between $s$ and $s'$ gets smaller. We show how a composition theorem for $\mathcal{P}$ allows one to trade off these two parameters: a tester for large $\varepsilon$ can be upgraded into one for small $\varepsilon$, at the price of larger gap between $s$ and~$s'$. The stronger the composition theorem, the more favorable this tradeoff is, and with an optimally strong composition theorem one is able to improve the $\varepsilon$-dependence without any associated price in the multiplicative gap between $s$ and $s'$:  
\medskip

\begin{tcolorbox}[colback = white,arc=1mm, boxrule=0.25mm]
\begin{theorem}[Composition theorems yield boosters for testers; informal]
Let $\mathcal{P} = \{ \mathcal{P}_s \}_{s\in\N}$ be a property  and $g : \bits^k\to\bits$ be such that $\mathcal{P}$ behaves linearly w.r.t.~$g$. Suppose that $\mathcal{P}$ admits an $(\eps_{\mathrm{small}}, \eps_{\mathrm{large}},\lambda)$-composition theorem w.r.t.~$g$.  Then any  $(\eps_{\mathrm{large}},ks,\lambda ks')$-tester for $\mathcal{P}$ can be converted in to an $(\eps_{\mathrm{small}}, s,s')$-tester for $\mathcal{P}$. 
\label{thm:generic boost intro}     
\end{theorem}
\end{tcolorbox}
\medskip 

We defer the precise definitions of the terms ``$(\eps_{\mathrm{small}},\eps_{\mathrm{large}},\lambda)$-composition theorem" and ``behaves linearly" to the body of the paper, mentioning for now that $\lambda \in [0,1]$ measures the strength of the composition theorem: such a theorem says that the composed function requires $\lambda k$ more resources to achieve $\eps_{\mathrm{large}}$ error than original function to achieve $\eps_{\mathrm{small}}$ error. Therefore $\lambda = \frac1{k}$ can be viewed as the threshold separating weak and strong composition theorems, with $\lambda = 1$ corresponding to an optimally strong one. (\Cref{cor:intro amplify junta complexity}, for example, achieves $\lambda = 0.99$.) Note that if $\lambda = 1$ in~\Cref{thm:generic boost intro}, then an $(\eps_{\mathrm{large}},s,s)$-tester for all $s$ yields an $(\eps_{\mathrm{small}},s,s)$-tester for all $s$. 

The formal version of~\Cref{thm:generic boost intro} will also show that it upgrades uniform-distribution testers to strong uniform-distribution testers, and distribution-free testers to strong distribution-free testers. This stands in contrast to standard boosting in learning which can only upgrade distribution-free learners.

\subsubsection{Example applications of~\Cref{thm:generic boost intro}: New implications for junta testing} 

As mentioned in the introduction, juntas are among the most basic and intensively-studied function classes in property testing. Owing to two decades of research, the complexity of testing juntas in the {\sl non-tolerant}  setting is now fairly well-understood: we have highly-efficient adaptive~\cite{Bla09}, non-adaptive~\cite{Bla08}, and distribution-free testers~\cite{LCSSYX18,Bsh19}, all of them achieving query complexities that are essentially optimal~\cite{CG04,Bla08,STW15,CSTWX18}.

  The picture is much less clear in the more challenging {\sl tolerant} setting. For the uniform distribution, the best known testers require exponentially many queries~\cite{BCELR18,ITW21}, and there are no known distribution-free testers. By generalization~\Cref{thm:generic boost intro} to the tolerant setting and instantiating it with our strong composition theorem for juntas, we obtain new implications, both positive and negative, that help clarify this picture. 

\paragraph{Positive implication: boosting of tolerant junta testers.} First, any tolerant junta tester for large distance parameter can now be converted into one for small distance parameters, at the price of a slight gap in the junta sizes of the Yes and No cases. For example, for both the uniform and distribution-free settings we get: 

\begin{corollary}[Boosting of tolerant junta testers; special case] Suppose we have a $\poly(r)$-query tester that distinguishes between 
\label{cor:boost juntas intro}
\begin{itemize} 
\item[$\circ$] Yes: $f$ is $\frac1{4}$-close to an $r$-junta 
\item[$\circ$] No: $f$ is $\frac1{3}$-far from every $r$-junta. 
\end{itemize} 
Then for every $\eps > 0$ we have a $\poly(r/\eps)$-query tester that distinguishes between 
\begin{itemize} 
\item[$\circ$] Yes: $f$ is $\eps$-close to an $r$-junta 
\item[$\circ$] No: $f$ is $\Omega(\eps)$-far from every $1.001r$-junta. 
\end{itemize} 
\end{corollary}

The resulting gap between the junta sizes of the Yes and No cases, while mild, is admittedly not ideal. As alluded to above, this stems from the fact that the ``strength parameter" of~\Cref{cor:intro amplify junta complexity} is $\lambda = 0.99$ and not $\lambda = 1$.  Designing boosters that do not incur this gap, either via an optimally strong composition theorem or otherwise, is a natural avenue for future work. 

On the other hand, we now show that even with this gap,~\Cref{cor:boost juntas intro} already carries with it an interesting consequence. This consequence crucially relies on our composition theorem for juntas being strong; the proof would not have gone through had the strength parameter of~\Cref{cor:intro amplify junta complexity} only been $\lambda = \frac1{k}$.

\paragraph{Negative implication: NP-hardness in the distribution-free setting.} This implication concerns the {\sl time} rather than query complexity of testers. The same proof of~\Cref{cor:boost juntas intro} also converts a $\poly(r,n)$-time tester into a $\poly(r,1/\eps,n)$-time tester. Implicit in the work of Hancock, Jiang, Li, and Tromp~\cite{HJLT96} is an NP-hardness result for tolerantly testing juntas in the distribution-free setting. One downside of their result is that it only holds in the regime of $\eps = 1/\poly(n)$. Applying the time-analogue of~\Cref{cor:boost juntas intro}, we lift this hardness up to the standard regime of constant $\eps$:  

\begin{corollary}[NP-hardness in the distribution-free testing]
\label{cor:NP hardness intro} 
The following task is NP-hard under randomized reductions. Given queries to a function $f : \bits^n\to\bits$, access to i.i.d.~draws from a distribution $\mathcal{D}$, and parameters $r\in \N$ and $\eps > 0$, distinguish between: 
\begin{itemize} 
\item[$\circ$] Yes: $f$ is $\frac1{4}$-close under $\mathcal{D}$ to an $r$-junta; 
\item[$\circ$] No: $f$ is $\frac1{3}$-far under  $\mathcal{D}$ from every $r$-junta. 
\end{itemize} 
\end{corollary} 

This implies a fairly dramatic separation  between the non-tolerant versus tolerant  versions of the problem. The recent $\poly(r)$-query non-tolerant testers~\cite{LCSSYX18,Bsh19} are also time efficient, running in $\poly(r,n)$ time.~\Cref{cor:NP hardness intro} shows that any tolerant tester, regardless of query efficiency, must have  time complexity that is as bad as that of SAT: e.g.~if SAT requires randomized exponential time, then so does any tolerant tester. 

In fact, our actual result is stronger than as stated in~\Cref{cor:NP hardness intro}: we prove that the task is NP-hard even if the Yes case states that $f$ is {\sl 0-close} under $\mathcal{D}$ to an $r$-junta. We therefore show  that the testers of~\cite{LCSSYX18,Bsh19} are quite fragile in the sense that they break if the Yes case in the definition of non-tolerant testing is changed from ``$f$ is an $r$-junta" to ``$f$ is $0$-close under $\mathcal{D}$ to an $r$-junta".  

\section{Other related work} 

\paragraph{O'Donnell's generalization of Yao's XOR lemma.}  
Yao's XOR lemma states that if $f$ is $\eps$-hard against circuits of size $s$, meaning every size-$s$ circuit differs from $f$ on at least an $\eps$-fraction of inputs, then $\mathrm{XOR}_k\circ f$ is $(\frac1{2} + \frac1{2}(1-2\eps)^k + \delta)$-hard against circuits of size $s'$ where 
\[ s'= \Theta\bigg(\frac{\delta^2}{\log(1/\eps)}\bigg)\cdot s. \]
The $(1-2\eps)^k$ term in the resulting inapproximability factor agrees precisely with the (univariate) noise stability of $\mathrm{XOR}_k$ at $\rho = 1-2\eps$.   In~\cite{OD02} O'Donnell showed that this is no coincidence. He proved a far-reaching generalization of Yao's XOR lemma that allows for an arbitrary combining function $g : \bits^k \to \bits$ instead of XOR, and showed that the resulting inapproximability of $g\circ f$ is given by the ``expected bias" of $g$, a quantity that is closely related to the (univariate) noise stability of $g$.  

Like Yao's XOR lemma,~\cite{OD02}'s  composition theorem is weak in the sense that the hardness of $g\circ f$ only holds against size $s'$ circuits where $s'  \ll s$.  (In fact,~\cite{OD02} incurs an additional multiplicative loss of $k$ in the resulting circuit size.) Our composition theorem concerns a different resource, juntas instead of circuits, and as emphasized in the introduction, our main focus is on proving a composition theorem that is strong in the sense of amplifying both the amount of resource required and the inapproximability factor. 

Both our work and~\cite{OD02} utilize  Fourier analysis in our proofs, which is to be expected given the centrality of noise stability to both works. That aside, our overall approach and techniques are entirely different from~\cite{OD02}'s---necessarily so, as we elaborate next.

\paragraph{Hardness amplification via boosting.} 
In~\cite{KS03} Klivans and Servedio observed that most known hardness amplification results are proved via a boosting-type argument. For example, for Yao's XOR lemma and~\cite{OD02}'s generalization of it, one proceeds by contradiction: one assumes that $\mathrm{XOR}_k\circ f$ can be mildly approximated by a size-$s'$ circuit $C$  (in the language of boosting, $C$ is a weak hypothesis for $\mathrm{XOR}_k \circ f$), and  one constructs a larger circuit $C^\star$ of size $s$ that well-approximates $f$ (i.e.~$C^\star$ is a strong hypothesis for $f$).  In boosting, the strong hypothesis is built out of many weak hypotheses; likewise, in Yao's XOR lemma the size-$s$ circuit $C^\star$ is built out of many size-$s'$ circuits that are like $C$. The work of~\cite{KS03} formalizes this connection. 

From this perspective, it becomes clear why such approaches are fundamentally limited to weak composition theorems where $s' \ll s$.  Strong composition theorems therefore necessitate a different tack, and indeed our proof proceeds via the forward implication instead of the contrapositive: we reason directly about the inapproximability of $g\circ f$ under the assumption about the inapproximability of $f$.  Somewhat ironically, our second main contribution is then an application {\sl of} strong composition theorems {\sl to} the boosting of property testers, which goes in the opposite direction to~\cite{KS03}'s ``Boosting $\Rightarrow$ Hardness Amplification" observation above.

\paragraph{Independent work of Chen and Patel \cite{CP23}.} A recent work of Chen and Patel also gives new lower bounds for tolerant junta testing. For the problem of testing whether an unknown function is $\eps_1$-close to or $\eps_2$-far from a $k$-junta under the {\sl uniform distribution}, they prove a {\sl query} lower bound of $k^{\Omega(\log(1/(\eps_2-\eps_1)))}$, which is  superpolynomial when the gap $\eps_2-\eps_1$ is {\sl subconstant}. This yields the first superpolynomial query complexity separation between tolerant and non-tolerant testing for a natural property of boolean functions. 

Their result is incomparable to \Cref{cor:NP hardness intro} in several respects. We give a {\sl time} lower bound when the gap $\eps_2-\eps_1$ is a {\sl fixed constant} in the {\sl distribution-free} setting. Being an NP-hardness result, our lower bound is conditional whereas theirs is unconditional. 

\section{Discussion and future work} 

Complexity measures can behave in highly counterintuitive ways under composition, which makes composition theorems, and strong composition theorems in particular,  tricky to prove.  
 A motivating goal of this work is to develop an  understanding of strong composition theorems from first principles, and hence our focus on junta complexity, perhaps the most basic complexity measure of a function. We are optimistic that our techniques can apply to other measures, though we believe that as in this work, much of the challenge will lie in first figuring out the right statement to prove. 
 
 Consider for example decision tree complexity, a natural next step from junta complexity. There are existing strong XOR lemmas for decision tree complexity, but they come with limitations and do not appear to be the final word. (Briefly, the XOR lemma of~\cite{Dru12} is only strong when the initial inapproximability factor $\eps_{\mathrm{small}}$ is at least a constant, and the strong XOR lemma of~\cite{BB19,BKLS20} only holds for decision trees that are allowed to ``abort".) Indeed, Shaltiel~\cite{Sha04} has shown that certain hoped-for strong XOR lemmas for decision tree complexity are false, though as he remarked, his counterexample ``seems to exploit defects in the formation of the problem rather than show that our general intuition for direct product assertions is false".  We hope that our results, and specifically the new connections to various notions of noise stability, can serve as a guide to the right statement for decision tree complexity and other measures. 

As for our second main result, the general connection between strong composition theorems and the boosting of property testers, we believe that it adds compelling algorithmic motivation to the study of composition theorems, a topic traditionally considered to be mostly of complexity-theoretic interest.  Likewise, we hope that our work spurs future research on this new notion of boosting for property testers, a notion that we believe is of interest  independent of the connections to composition theorems. For example, an ambitious goal for future work is to broadly understand when and how a tester for constant distance parameter $\eps$ can be automatically upgraded into one with the optimal $\eps$-dependence, as well as the associated costs of such a transformation.

%% file: Preliminaries.tex
\section{Preliminaries}

\pparagraph{Distributions and random variables.} We use \textbf{bold font} (e.g $\bx \sim \mcD$) to denote random variables. 
For any set $S$, we use $\bx \sim S$ as shorthand for $\bx \sim \mathrm{Unif}(S)$ where $\mathrm{Unif}(\cdot)$ denotes the uniform distribution. Of particular importance to this work will be $\mu$-biased distributions over the Boolean hypercube.
\begin{definition}[$\mu$-biased distribution]
    \label{def:mu-biased}
    For any $\mu \in (-1,1)$, we use $\pi_\mu$ to denote the unique distribution over $\bits$ with mean $\mu$. Formally, for $\by \sim \pi_{\mu}$,
    \begin{equation*}
        \by = \begin{cases}
            1 &\text{with probability } \frac{1 + \mu}{2} \\
            -1 &\text{with probability } \frac{1 - \mu}{2}.
        \end{cases}
    \end{equation*} 
    Similarly, for ${\vrho} \in [-1,1]^k$, we use $\pi_{\vrho}$ to denote the product distribution $\pi_{\vrho_1} \times \cdots \times \pi_{\vrho_k}$.
\end{definition}
\begin{definition}[$\vrho$-correlated]
    Fix some bias $\mu \in (-1,1)$. For any $\vrho \in [0,1]^k$ and $y \in \bits^k$, we write $\bz \rhosim y$ to denote that for each $i \in [k]$, $\bz_i$ is independently drawn as
    \begin{equation*}
        \bz_i = \begin{cases}
            y_i &\text{with probability }\vrho_i\\
            \text{Drawn from $\pi_{\mu}$} &\text{with probability }1 - \vrho_i.
        \end{cases}
    \end{equation*}
\end{definition}
Whenever we use the above notation, the choice of $\mu$ will be clear from context. This gives the following more succinct way to express \Cref{def:multi-noise}, defining multivariate noise stability,
\begin{equation*}
    \Stab_{\mu,\vrho}(g) \coloneqq \Ex_{\by \sim (\pi_{\mu})^k,\bz \rhosim \by}[g(\by)g(\bz)].
\end{equation*}

\pparagraph{Some useful sets.} For any integers $a \leq b$, we use $[a,b]$ as shorthand for the set $\{a, a+1, \ldots, b\}$. Similarly, for $b \geq 1$, we use $[b]$ as shorthand for the set $[1,b]$. For any set $S$ and $\ell \leq |S|$, we use $\binom{S}{\ell}$ to denote all subsets of $S$ with cardinality $\ell$.

\pparagraph{Junta complexity.} For any function $f: \bits^n \to \bits$, and $S \subseteq [n]$, we say that $f$ is an $S$-junta if for all $x,y \in \bits^n$ for which $x_i = y_i$ whenever $i \in S$ it holds that $f(x) = f(y)$. With a slight abuse of notation, when $r \in [n]$ is an integer, we say that $f$ is an $r$-junta if there is a set $|S| \leq r$ for which $f$ is an $r$-junta.

\pparagraph{Advantage.} 
For any functions $f, g:\bits^n \to \bits$ and distribution $\mcD$ over $\bits^n$, we define
\begin{equation*}
    \adv_\mcD(f,g) \coloneqq \Ex_{\bx \sim \mcD}[f(\bx) g(\bx)].
\end{equation*}
With a slight abuse of notation, we define for $f:\bits^n \to \bits$ and $S \subseteq [n]$,
\begin{equation*}
    \adv_{\mcD}(f,S) \coloneqq \max_{S\text{-junta }g:\bits^n \to \bits} \adv_{\mcD}(f,g).
\end{equation*}
Similarly, for $r \in [n]$,
\begin{equation*}
    \adv_{\mcD}(f,r) \coloneqq \max_{r\text{-junta }g:\bits^n \to \bits} \adv_{\mcD}(f,g).
\end{equation*}
When the base distribution $\mcD$ is clear, we will drop it from our notation. Furthermore, for any function $f:\bits^n \to \bits$ and $S \subseteq [n]$ or $r \in [n]$, we use $\tilde{f}_S$ and $\tilde{f}_r$ to denote the $S$-junta and $r$-junta respectively maximizing the above two advantages.

\pparagraph{Function composition.}
For a function $f: \bits^n \to \bits$, its direct product $f^{\otimes k}:\paren*{\bits^{n}}^{k} \to \bits^k$ is defined as
\begin{equation*}
        f^{\otimes k}(x^{(1)}, \ldots, x^{(k)}) = (f(x^{(1)}), \ldots, f(x^{(k)})).
\end{equation*}
For any $g:\bits^k \to \bits$, we use $g \circ f:\paren*{\bits^{n}}^{k} \to \bits$ as shorthand for $g\circ f^{\otimes k}$, meaning,
\begin{equation*}
    (g\circ f)(x^{(1)}, \ldots, x^{(k)}) = g(f(x^{(1)}), \ldots, f(x^{(k)})).
\end{equation*}

\pparagraph{Vector powers.} For any vector $v \in \R^k$ and set $S \subseteq [k]$, we'll use the notation $v^S$ as shorthand for
\begin{equation*}
    v^S \coloneqq \prod_{i \in S} v_i.
\end{equation*}

\subsection{Fourier Analysis}

Our proof of \Cref{thm:composition theorem intro 1} will make heavy use of Fourier analysis over the $\mu$-biased hypercube, $(\pi_\mu)^k$. In this section, we will review relevant definitions and facts. A more complete exposition is given in \cite{ODBook}.

For any $\mu \in (-1,1)$, we define $\phi_\mu(x) \coloneqq \frac{x-\mu}{\sigma}$ where $\sigma \coloneqq \sqrt{1 - \mu^2}$. Every $g: \bits^k \to \R$ can be uniquely decomposed as
\begin{equation*}
    g(y) = \sum_{S \subseteq [k]} \hat{g}_{\mu}(S) \prod_{i \in S}\phi_\mu(y_i)\quad\quad\text{where}\quad\quad \hat{g}_{\mu}(S) = \Ex_{\by \sim (\pi_\mu)^k}\bracket*{g(\by) \prod_{i \in S} \phi_\mu(\by_i)}.
\end{equation*}
This decomposition has a number of useful properties stemming from the fact that transforming $g$ from its representation as a truth table to its Fourier coefficients $\hat{g}_{\mu}(S)$ is an orthonormal transformation.
\begin{fact}[Basic facts about the Fourier decomposition]
    \label{fact:pp-fourier}
    \leavevmode
    \begin{enumerate}
        \item \textnormal{\textbf{Plancherel's theorem:}} For any $g, h: \bits^k \to \R$ and $\mu \in (-1,1)$,
        \begin{equation*}
            \Ex_{\by \sim (\pi_\mu)^k}[g(\by)h(\by)] = \sum_{S \subseteq [k]} \hat{g}_\mu(S)\hat{h}_\mu(S).
        \end{equation*}
        \item \textnormal{\textbf{Parseval's theorem:}} For any $g: \bits^k \to \R$ and $\mu \in (-1,1)$,
        \begin{equation*}
            \Ex_{\by \sim (\pi_\mu)^k}[g(\by)^2] = \sum_{S \subseteq [k]} \hat{g}_\mu(S)^2.
        \end{equation*}
    \end{enumerate}
\end{fact}

In particular, when $g$ has a range of $\bits$, Parseval's theorem guarantees that the sum of its squared Fourier coefficients is $1$. As a result, the following distribution is well defined.
\begin{definition}[Spectral sample]
    For any $g: \bits^k \to \bits$ and bias $\mu \in (-1,1)$, the \emph{spectral sample} of $g$, denoted $\mcS_\mu(g)$, is the probably distribution over subsets of $[k]$ in which the set $S$ has probability $\hat{g}_\mu(S)^2$.
\end{definition}

The Fourier decomposition gives a concise way to represent important quantities, as in the following results.
\begin{proposition}[Multivariate noise stability from the Fourier spectrum.]
    \label{prop:stab-fourier}
    For any $\mu \in (-1,1)$ and $\vrho \in [0,1]^k$, $\Stab_{\mu, \vrho}$ can be related to $g$'s $\mu$-biased Fourier decomposition as,
    \begin{equation*}
        \Stab_{\mu, \vrho}(g) = \sum_{S \subseteq [k]} \hat{g}(S)^2 \vrho^S = \Ex_{\bS \sim \mcS_{\mu}(g)}[(\vrho)^{\bS}].
    \end{equation*}
\end{proposition}
\begin{proof}
    We define $g^{(\vrho)}(y) \coloneqq \Ex_{\bz \rhosim y}[g(\bz)]$. Then, by Plancherel's theorem,
    \begin{equation*}
        \Stab_{\mu, \vrho}(g) = \Ex_{\by \sim (\pi_{\mu})^k}[g(\by) g^{(\vrho)}(\by)] = \sum_{S \subseteq [k]} \wh{g}_\mu(S) \wh{g^{(\vrho)}}_\mu(S).
    \end{equation*}
    Next, we compute the Fourier decomposition of $g^{(\vrho)}$.
    \begin{align*}
        \wh{g^{(\vrho)}}_\mu(S) &= \Ex_{\by \sim (\pi_{\mu})^k}\bracket*{g^{(\vrho)}(\by) \prod_{i \in S} \phi_\mu(\by_i)}\\
        &= \Ex_{\by \sim (\pi_{\mu})^k, \bz \rhosim \bx}\bracket*{g(\bz) \prod_{i \in S} \phi_\mu(\by_i)}\\
        &= \Ex_{\by \sim (\pi_{\mu})^k, \bz \rhosim \bx}\bracket*{g(\by) \prod_{i \in S} \phi_\mu(\bz_i)} \tag{$(\by,\bz)$ distributed identically to $(\bz, \by)$}\\
        &= \Ex_{\by\sim (\pi_\mu)^k}\bracket*{g(\by) \cdot \Ex_{\bz \rhosim \by}\bracket*{ \prod_{i \in S} \phi_\mu(\bz_i)}}.
    \end{align*}
    Applying the independence of $\bz_1, \ldots, \bz_k$ conditioned on $\by$ and that $\Ex[\phi_\mu(\bz_i)] = \vrho_i \phi_\mu(\by_i)$,
    \begin{align*}
         \wh{g^{(\vrho)}}_\mu(S) &= \Ex_{\by\sim (\pi_\mu)^k}\bracket*{g(\by) \cdot \prod_{i \in S} \vrho_i \phi_\mu(\by_i)} \\
        &= (\vrho)^S \cdot\Ex_{\by\sim (\pi_\mu)^k}\bracket*{g(\by) \cdot \prod_{i \in S}\phi_\mu(\by_i)} = (\vrho)^S \wh{g}_\mu(S).
    \end{align*}
    Putting the above together,
    \begin{equation*}
         \Stab_{\mu, \vrho}(g) = \sum_{S \subseteq [k]}\hat{g}_\mu(S)^2 (\vrho)^S. \qedhere
    \end{equation*}
\end{proof}
One immediate corollary of the above is that multivariate noise stability is monotone.
\begin{corollary}[Multivariate noise stability is monotone]
    \label{cor:multi-monotone}
    For any $\mu \in (-1,1)$, $g:\bits^k \to \bits$, and $\vrho, \vec{\rho'} \in [0,1]^k$ satisfying $\vrho_i \leq \vec{\rho'}_i$ for all $i \in [k]$,
    \[\stab_{\mu, \vrho}(g) \leq \stab_{\mu, \vec{\rho'}}(g).\]
\end{corollary}

Recall that for any $\nu \in [-1,1]^k$, the distribution $\pi_{\nu}$ is the unique product distribution supported on $\bits^k$ with mean $\nu$. The Fourier decomposition of $g$ also gives a useful way to compute $\Ex_{\by \sim \pi_{\nu}}[g(\by)]$.
\begin{fact}
    \label{fact:fourier-sample}
    For any $g: \bits^k \to \R$, $\mu \in (-1,1)$, and $\nu \in [-1,1]^k$,
    \begin{equation*}
        \Ex_{\by \sim \pi_{\nu}}[g(\by)] = \sum_{S \subseteq [k]}\hat{g}_\mu(S) \prod_{i \in S}\phi_\mu(\nu_i).
    \end{equation*}
\end{fact}
\begin{proof}
    We expand $g$ into it's Fourier decomposition
    \begin{align*}
        \Ex[g(\by)] &= \sum_{S \subseteq [k]} \hat{g}_\mu(S) \Ex\bracket*{\prod_{i \in S}\phi_\mu(\by_i)} \tag{Linearity of expectation} \\
        &= \sum_{S \subseteq [k]} \hat{g}_\mu(S)\prod_{i \in S} \Ex\bracket*{\phi_\mu(\by_i)} \tag{$\by_1, \ldots, \by_k$ are independent} \\
        &=  \sum_{S \subseteq [k]} \hat{g}_\mu(S)\prod_{i \in S} \Ex\bracket*{\frac{\by_i - \mu}{\sigma}} \tag{Definition of $\phi_\mu$}\\
        &=  \sum_{S \subseteq [k]} \hat{g}_\mu(S)\prod_{i \in S} \phi_\mu(\nu_i). \tag{Linearity of expectation}
    \end{align*}
\end{proof}

%% file: StrongDirectProduct.tex
\section{A strong composition theorem for juntas}

In this section, we characterize the junta size required to approximate $g \circ f$ in terms of the \emph{multivariate noise stability} of $g$, and the junta size required to approximate $f$.
\begin{theorem}[A strong composition theorem for junta complexity, generalization of \Cref{thm:composition theorem intro 1}]
    \label{thm:strong-formal}
    For any $g: \bits^k \to \bits$, $f: \bits^n \to \bits$ and base distribution $\mcD$ over $\bits^n$, let $\mu = \Ex_{\bx\sim \mcD}[f(\bx)]$.
    \begin{enumerate}
        \item \textbf{Lower bound on advantage:} For any approximators $q^{(1)}, \ldots, q^{(k)}: \bits^n \to \bits$, define the \emph{lower normalized correlations}, for each $i \in [k]$ as
        \begin{equation*}
            \alpha_i \coloneqq \max\paren*{0, \frac{\adv_{\mcD}(f, q^{(i)})^2 - \mu^2}{1 - \mu^2}}.
        \end{equation*}
        Then, there is an $h:\bits^k \to \bits$ for which
        \begin{equation*}
            \adv_{\mcD^k}(g\circ f, h (q^{(1)}, \ldots, q^{(k)})) \geq \stab_{\mu, \alpha}(g).
        \end{equation*}
        \item \textbf{Upper bound on advantage:} For any $S_1,\ldots, S_k$, define the \emph{upper normalized correlation} as
        \begin{equation*}
            \beta_i \coloneqq \max\paren*{0,\frac{\adv_\mcD(f, S_i) - \mu^2}{1 - \mu^2}},
        \end{equation*}
        construct $S \subseteq [n] \times [k]$ by taking $S_1$ from the first block, $S_2$ from the second block, and so on (formally $S \coloneqq \cup_{i \in [k], j \in S_i}\{(j,i)\}$). Then, 
        \begin{equation*}
            \adv_{\mcD^k}(g\circ f, S) \leq \sqrt{\stab_{\mu, \beta}(g)}.
        \end{equation*}
    \end{enumerate}
\end{theorem}
Our goal is to understand the error of the best $R$-junta approximating $g \circ f$. \Cref{thm:strong-formal} says that for \emph{any} way to partition $R = r_1 + \cdots r_k$, the approximator $h (\tilde{f}_{r_1}, \ldots, \tilde{f}_{r_k})$ achieves nearly optimal advantage across all $R$-juntas that partition their budget this way. Of course, by maximizing both sides across all partitions, we can conclude that there is some partitioning and function $h$ for which $h (\tilde{f}_{r_1}, \ldots, \tilde{f}_{r_k})$ has nearly optimal advantage among all $R$-juntas. Indeed, as a simple corollary of \Cref{thm:strong-formal}, we can show that the \emph{error} of the optimal canonical composed form approximator is within a factor of $4$ of the optimal approximator. Recall that $\error_{\mcD}(q_1,q_2) = \Pr_{\bx \sim \mcD}[q_1(\bx) \neq q_2(\bx)]$ and is related to advantage via the equality $\adv = 1 - 2\cdot \error$.
\begin{corollary}[Canonical composed form achieves nearly optimal error]
    \label{cor:error-4}
    For any $g: \bits^k \to \bits, f:\bits^n \to \bits$, junta budget $R$, and base distribution $\mcD$, there is an $h:\bits^n \to \bits$ and partition of the budget $r_1 + \cdots + r_k = R$ for which,.
    \begin{equation*}
        \error_{\mcD^k}(g\circ f, h (\tilde{f}_{r_1}, \ldots, \tilde{f}_{r_k})) \leq 4 \cdot \error_{\mcD^k}(g\circ f, R).
    \end{equation*}
\end{corollary}
When $\mu = 0$, the guarantee of \Cref{thm:strong-formal} can further be given in the concise form of \Cref{thm:composition theorem intro 1}: For an appropriately chosen $\vrho \in [0,1]^k$,
\begin{align*}
    \Stab_{\vec{\rho}}(g)^2 &\le  \textnormal{Advantage of optimal canonical composed form approximator} \\
    &\le \textnormal{Advantage of optimal approximator} \le \sqrt{\Stab_{\vec{\rho}}(g)}.
\end{align*}
We include the proofs of \Cref{cor:error-4} and \Cref{thm:composition theorem intro 1} in \Cref{subsec:corollary-proofs}.

\subsection{Proof of the lower bound on advantage} 
\label{sec:proof of lower bound on advantage} 

In this subsection, we show that  $(x_1, \ldots, x_k) \to h(\tilde{f}_{r_1}(x_1), \ldots, \tilde{f}_{r_k}(x_k))$ is close to the best $R$-junta approximator for $g \circ f$. Here, the function $h$ can be different than $g$, and this is necessary as shown in the counterexample to conjecture 2 in \Cref{subsec:counterexamples}.

\begin{lemma}[Part 1 of \Cref{thm:strong-formal}: Lower bound on advantage]
    \label{lem:upper-junta}
    For any $g:\bits^k \to \bits$, $f:\bits^n \to \bits$, and approximators $q^{(1)}, \ldots, q^{(k)}$, there is some $h:\bits^k \to \bits$ for which
    \begin{equation*}
        \adv_{\mcD^k}(g\circ f, h \circ (q^{(1)}, \ldots, q^{(k)})) \geq \stab_{\mu, \alpha}(g),
    \end{equation*}
    where $\mu = \Ex_{\bx \sim \mcD}[f(\bx)]$ and for each $i \in [k]$,
    \begin{equation*}
        \alpha_i \coloneqq \max\paren*{0, \frac{\adv(f, q^{(i)})^2 - \mu^2}{1 - \mu^2}}.
    \end{equation*}
\end{lemma}
 
Note $\alpha_i$ naturally interpolates between $0$ and $1$. Setting $q^{(i)}$ to the better of the constant $-1$ or the constant $+1$ function will lead to $\alpha_i = 0$, while setting $q^{(i)} = f$ gives $\alpha_i = 1$.

\subsubsection{Characterizing the advantage of composed form approximators} 

To ease notation, we begin with a simpler setting. Suppose we use the same budget, $r \coloneqq R/k$, in each of the $k$ pieces. Our goal is to understand
\begin{equation*}
    \max_{h:\bits^k \to \bits} \adv(g\circ f, h\circ \tilde{f}_{r})
\end{equation*}
in terms of the noise sensitivity of $g$ and $\adv(f, \tilde{f}_r)$. To do so, we will consider \emph{unbalanced} noise stability.
\begin{definition}[Unbalanced noise stability]
    For any $x \in \bits^k$, we use the notation $\by \absim x$ to denote that for each $i \in [k]$, $\by_i$ is independently drawn as
    \begin{enumerate}
        \item If $x_i = -1$, with probability $a$, we set $\by_i = x_i$ and otherwise set $\by_i = -x_i$
        \item If $x_i = 1$, with probability $b$, we set $\by_i = x_i$ and otherwise set $\by_i = -x_i$.
    \end{enumerate}
    For any $g,h:\bits^k \to \bits$, $\mu \in [-1,1]$ and $a,b \in [0,1]$, we define the \emph{unbalanced} noise stability as 
    \begin{equation*}
        \unstab_{\mu, (a,b)}(g,h) = \Ex_{\bx \sim (\pi_\mu)^k, \by \absim \bx}[g(\bx)h(\by)].
    \end{equation*}
\end{definition}
We refer to the above notion as \emph{unbalanced} because when drawing $\by \absim x$, the probability of the $i^{\text{th}}$ coordinate flipping from $-1$ to $1$ and from $1$ to $-1$ may differ. Unbalanced noise stability is useful in our setting due to the following proposition.
\begin{proposition}
    For any $f, \tilde{f}: \bits^n \to \bits$ and $g,h:\bits^k \to \bits$,
    \begin{equation*}
        \Ex_{\bx \sim \mcD^k}[(g \circ f)(\bx) \cdot (h \circ \tilde{f})(\bx)] = \unstab_{\mu, (a,b)}(g,h),
    \end{equation*}
    where
    \begin{equation*}
        \mu \coloneqq \Ex_{\bx \sim \mcD}[f(\bx)],\quad\quad
        a \coloneqq \Prx_{\bx \sim \mcD}[\tilde{f}(\bx) = -1 \mid f(\bx) = -1],\quad\quad
        b \coloneqq \Prx_{\bx \sim \mcD}[\tilde{f}(\bx) = 1 \mid f(\bx) = 1].
    \end{equation*}
\end{proposition}
\begin{proof}
    Draw $\bx \sim\mcD^k$ and then define $\by \coloneqq f^{\otimes k}(\bx), \tilde{\by} \coloneqq \tilde{f}^{\otimes k}(\bx)$. Clearly,
    \begin{equation*}
        \Ex_{\bx \sim \mcD^k}[(g \circ f)(\bx) \cdot (h \circ \tilde{f})(\bx)] = \Ex[g(\by) h(\tilde{\by})].
    \end{equation*}
    Furthermore, the distribution of $\by, \tilde{\by}$ is equivalent to if we drew $\by \sim (\pi_\mu)^k, \tilde{\by} \absim \by$. The above quantity therefore matches the definition of $\unstab_{\mu, (a,b)}(g,h)$.
\end{proof}

\subsubsection{Unbalanced noise stability behaves strangely}

The most basic requirement of our approximation for $g \circ f$ is that it have advantage at least $0$, as either the constant $-1$ or the constant $+1$ function is guaranteed to have such an advantage. Indeed, in the balanced case, it is well known that the approximation will satisfy this basic requirement even if we take $h = g$. 
\begin{fact}
    For any $g:\bits^k \to \bits$ and $a \in [0,1/2]$,
    \begin{equation*}
        \unstab_{0, (a,a)}(g,g) \geq 0.
    \end{equation*}
\end{fact}
However, in the unbalanced case, this basic requirement no longer holds.
\begin{proposition}
    \label{prop:unbal-weird}
    For any $k \geq 0$, and $a,b \in [0,1]$ for which $|a-b| \geq 0.01$, there is a function $g:\bits^k \to \bits$ for which
    \begin{equation*}
        \unstab_{0, (a,b)}(g,g) \leq -(1-2^{-\Omega(k)}).
    \end{equation*}
\end{proposition}
\begin{proof}
    Without loss of generality, we assume $b \geq a  + 0.01$. We define
    \begin{equation*}
        g(x) \coloneqq \begin{cases}
            1 &\text{if }\sum_{i \in [k]}x_i \geq 0.005k,\\
            -1&\text{otherwise.}
        \end{cases}
    \end{equation*}
    Draw $\bx \sim (\pi_\mu)^k, \by \absim \bx$. Then,
    \begin{equation*}
        \Ex\bracket*{\sum_{i \in [k]}\bx_i} = 0 ,\quad\quad\quad\quad \Ex\bracket*{\sum_{i \in [k]}\by_i} = k(b-a).
    \end{equation*}
    Furthermore, a standard application of Hoeffding's inequality implies that
    \begin{equation*}
        \Prx[g(\bx) = 1] \leq 2^{-\Omega(k)} ,\quad\quad\quad\quad \Prx[g(\by) = -1]\leq 2^{-\Omega(k)}.
    \end{equation*}
    By union bound, with probability at least $2^{-\Omega(k)}$, we have that both $g(\bx) = -1$ and $g(\by) = 1$. This implies the desired result.
\end{proof}

\subsubsection{Unbalanced noise stability behaves well if we use the best $h$}
Surprisingly, we show that if we use the best $h$, our approximation does meet this most basic requirement. Furthermore, we can relate it to the classical notion of \emph{balanced} noise stability.   The below Lemma directly implies \Cref{lem:upper-junta}.

\begin{lemma}
    \label{lem:unbal-best-h}
    For any $g:\bits^k \to \bits$ and distribution over $\bx, \by$ each in $\bits^k$ satisfying,
    \begin{enumerate}
        \item The pairs $(\bx_1, \by_1), \ldots, (\bx_k, \by_k)$ are independent of one another.
        \item The means satisfy $\Ex[\bx_1] = \cdots = \Ex[\bx_k] = \mu$.
    \end{enumerate}
    Define the correlations $\alpha_1, \ldots, \alpha_k$ as
    \begin{equation*}
        \alpha_i \coloneqq \max\paren*{0,\frac{\Ex[\bx_i \by_i]^2 - \mu^2}{1 - \mu^2}}.
    \end{equation*}
    Then, there is an $h:\bits^k \to \bits$ for which
    \begin{equation*}
        \Ex[g(\bx)h(\by)] \geq \stab_{\mu, \alpha}(g).
    \end{equation*}
\end{lemma}

Comparing to \Cref{prop:unbal-weird}, if $\mu = 0$, then $\alpha_i = \max(0,1-a-b)$ for all $i \in [k]$. Since $\stab_{\mu, \alpha}(g) \geq 0$ whenever $\alpha \geq 0$, \Cref{lem:unbal-best-h} shows that the phenomenon in \Cref{prop:unbal-weird} cannot occur if we use the best approximator $h$.

The following Lemma will be useful in the proof of \Cref{lem:unbal-best-h}.

\begin{lemma}
    \label{lem:avg-to-stab}
    For any function $g: \bits^k \to \bits$, let $\bnu_1, \ldots, \bnu_k$ be independent random variables each with mean $\mu$ and supported on $[-1,1]$. Then,
    \begin{equation*}
        \Ex_{\bnu}\bracket*{\Ex_{\by \sim \pi_{\bnu}}[g(\by)]^2} = \Stab_{\mu, (\Ex[\phi_\mu(\bnu_1)^2], 
        \ldots, \Ex[\phi_\mu(\bnu_k)^2])}(g).
    \end{equation*}
\end{lemma}
\begin{proof}
    We'll use the $\mu$-biased Fourier expansion of $g$. Applying \Cref{fact:fourier-sample},
    \begin{align*}
        \Ex_{\bnu}\bracket*{\Ex_{\by \sim \pi_{\bnu}}[g(\by)]^2} &= \Ex_{\bnu}\bracket*{\paren*{\sum_{S \subseteq [k]}\hat{g}(S) \prod_{i \in S}\phi_\mu(\bnu_i)}^2} \\
        &= \sum_{S_1, S_2 \subseteq [k]}\hat{g}(S_1)\hat{g}(S_2)\Ex\bracket*{\prod_{i \in S_1}\phi_\mu(\bnu_i)\prod_{i \in S_2}\phi_\mu(\bnu_i)}.
    \end{align*}
    We claim that, in the above sum, any term in which $S_1 \neq S_2$ is equal to $0$. Let $S_1 \triangle S_2$ denote the symmetric difference of $S_1$ and $S_2$. Then, due to the independence of $\bnu_1, \ldots, \bnu_k$,
    \begin{equation*}
        \Ex\bracket*{\prod_{i \in S_1}\phi_\mu(\bnu_i)\prod_{i \in S_2}\phi_\mu(\bnu_i)} = \prod_{i \in S_1 \cap S_2} \Ex[\phi_\mu(\bnu_i)^2] \prod_{i \in S_1 \triangle S_2} \Ex[\phi_\mu(\bnu_i)].
    \end{equation*}
    Since the mean of $\bnu_i$ is $\mu$, $\Ex[\phi_\mu(\bnu_i)] = \phi_\mu(\mu) = 0$. If $S_1 \neq S_2$, there is at least one element in $S_1 \triangle S_2$, and so the term is $0$. We are therefore left with,
    \begin{align*}
        \Ex_{\bnu}\bracket*{\Ex_{\by \sim \mcD(\bnu)}[g(\by)]^2} &= \sum_{S \subseteq [k]}\hat{g}(S)^2\prod_{i \in S}\Ex\bracket*{\phi_\mu(\bnu_i)^2}.
    \end{align*}
    This is exactly the Fourier expansion for the claimed result.
\end{proof}

We'll also use the following proposition.
\begin{proposition}
    \label{prop:adv-to-square}
    For any random variable $\bnu$ bounded on $[-1,1]$ almost surely and with mean $\mu$,
    \begin{equation*}
        \max\paren*{0,\frac{\Ex[\abs{\bnu}]^2 - \mu^2}{1 -  \mu^2}} \leq\Ex[\phi_\mu(\bnu)^2] \leq \frac{\Ex[\abs{\bnu}] - \mu^2}{1 -  \mu^2} .
    \end{equation*}
\end{proposition}
\begin{proof}
    We expand, using linearity of expectation,
    \begin{equation*}
        \Ex[\phi_\mu(\bnu)^2] = \Ex\bracket*{\frac{(\bnu - \mu)^2}{1 - \mu^2}} = \frac{\Ex[\rho^2] - 2\mu\Ex[\bnu] + \mu^2}{1 - \mu^2}.
    \end{equation*}
    Since $\Ex[\bnu] = \mu$, we have that $\Ex[\phi_\mu(\bnu)^2] = \frac{\Ex[\bnu^2] -  \mu^2}{1 - \mu^2}$. Therefore, by Jensen's inequality,
    \begin{equation*}
        \frac{\Ex[\abs{\bnu}]^2 - \mu^2}{1 -  \mu^2} \leq\Ex[\phi_\mu(\bnu)^2].
    \end{equation*}
    Furthermore, since $\bnu^2 \leq \abs{\bnu}$,
    \begin{equation*}
        \Ex[\phi_\mu(\bnu)^2] \leq \frac{\Ex[\abs{\bnu}] - \mu^2}{1 -  \mu^2}.
    \end{equation*}
    Lastly, $\Ex[\phi_\mu(\bnu)^2] \geq 0$ follows from non-negativity.
\end{proof}

Finally, we are ready to prove \Cref{lem:unbal-best-h}.
\begin{proof}[Proof of \Cref{lem:unbal-best-h}]
    For any $y \in \bits^n$, we define
    \begin{equation*}
        g_{\avg}(y) = \Ex[g(\bx) \mid \by = y].
    \end{equation*}
    Then, setting $h(y) \coloneqq \sign(g_{\avg}(y))$,
    \begin{equation*}
        \Ex[g(\bx)h(\by)] = \Ex_{\by}\bracket*{\abs*{g_{\avg}(\by)}} \geq \Ex_{\by}\bracket*{\paren*{g_{\avg}(\by)}^2}.
    \end{equation*}
    Note that, conditioning on $\by = y$, the distribution of $\bx$ is still product. Let $\nu(y)$ be the mean of this distribution, so that
    \begin{equation*}
        g_{\avg}(y) = \Ex_{\bz \sim \pi_{\nu(y)}}\bracket*{g(\bz)}.
    \end{equation*}
    By \Cref{lem:avg-to-stab},
    \begin{equation*}
        \Ex_{\by}\bracket*{\paren*{\Ex_{\bz \sim \pi_{\nu(\by)}}\bracket*{g(\bz)}}^2} = \stab_{\mu, (\Ex[\phi_{\mu}(\nu(\by)_1)^2], \ldots, \Ex[\phi_{\mu}(\nu(\by)_k)^2]}(g).
    \end{equation*}
    For each $i \in [k]$,
    \begin{align*}
        \Ex[\phi_{\mu}(\nu(\by)_i)^2] &\geq  \max\paren*{0,\frac{\Ex_{\by}[\abs{\nu(\by)_i}]^2 - \mu^2}{1 -  \mu^2}}  \tag{ \Cref{prop:adv-to-square}}\\
        &\geq   \max\paren*{0,\frac{\Ex_{\by}[\by_i\nu(\by)_i]^2 - \mu^2}{1 -  \mu^2}} \tag{$\abs{x} \geq cx$ when $c \in \bits$} \\
        &=   \max\paren*{0,\frac{\Ex_{\bx,\by}[\by_i\bx_i]^2 - \mu^2}{1 -  \mu^2}} \tag{Definition of $\nu(y)$} \\
        & = \alpha_i.
    \end{align*}
    Putting all of the above together,
    \begin{align*}
        \Ex[g(\bx)h(\by)] &\geq \stab_{\mu, (\Ex[\phi_{\mu}(\nu(\by)_1)^2], \ldots, \Ex[\phi_{\mu}(\nu(\by)_k)^2]}(g) \\
        &\geq \stab_{\mu, \rho}(g),
    \end{align*}
    where the final inequality follows from the monotonicity of noise stability.
\end{proof}

\subsection{Proof of the upper bound on advantage} 

In this section, we prove the following.

\begin{lemma}[Part 2 of \Cref{thm:strong-formal}: Upper bound on advantage]
    \label{lem:lower}
    For any $g: \bits^k\to\bits$, $f:\bits^n \to \bits$, $\mu \coloneqq \Ex_{\bx \sim \mcD}[f(\bx)]$, and $S_1,\ldots, S_k$, define the \emph{upper normalized correlation} as
    \begin{equation*}
        \beta_i \coloneqq \frac{\adv_{\mcD}(f, S_i) - \mu^2}{1 - \mu^2}.
    \end{equation*}
    For $S \subseteq [n] \times [k]$ constructed by taking $S_1$ from the first block, $S_2$ from the second block, and so on (formally $S \coloneqq \cup_{i \in [k], j \in S_i}\{(j,i)\}$).. Then, 
    \begin{equation*}
        \adv_{\mcD^k}(g\circ f, S) \leq \sqrt{\stab_{\mu, \beta}(g)}.
    \end{equation*}
\end{lemma}

To begin with, we rewrite advantage in the following form.
\begin{proposition}
    \label{prop:g-avg}
    For any function $q: \bits^m \to \bits$, distribution $\mcD$ over $\bits^m$, and $S \subseteq [m]$, define
    \begin{equation}
        \label{eq:def-avg}
        q_{S, \mcD}^{\avg}(x) \coloneqq \Ex_{\by \sim \mcD}[q(\by) \mid \by_S = x_S],
    \end{equation}
    where $y_S = x_S$ is shorthand for $x_i = y_i$ for all $i \in S$. Then,
    \begin{equation*}
        \Adv_{\mcD}(q, S) = \Ex_{\bx \sim \mcD}\bracket*{\abs*{q_{S, \mcD}^{\avg}(\bx)}}.
    \end{equation*}
\end{proposition}
\begin{proof}
    Consider any $S$-junta $h$. Then,
    \begin{equation*}
        \adv_{\mcD}(q, h) = \Ex_{\by \sim \mcD}[
        q(\by) h(\by)] = \Ex_{\bx \sim \mcD}\bracket*{\Ex_{\by \sim \mcD}[q(\by) h(\by) \mid \bx_S = \by_S]}.
    \end{equation*}
    Since $h$ is an $S$-junta, it must classify $x$ and $y$ the same whenever $x_S = y_S$. Therefore,
    \begin{align*}
        \adv(q, h) &= \Ex_{\bx \sim \mcD}\bracket*{h(\bx)\Ex_{\by \sim \mcD}[q(\by) \mid \bx_S = \by_S]}\\
        &= \Ex_{\bx \sim \mcD}\bracket*{h(\bx)q^{\avg}_{S,\mcD}(\bx)}.
    \end{align*}
    to maximize the above advantage among all $h$, we set $h(x) = \sign(q^{\avg}_{S, \mcD}(x))$, in which case
    \begin{equation*}
        \adv(q, h) = \Ex_{\bx \sim \mcD}\bracket*{\abs*{q^{\avg}_{S, \mcD}(\bx)}}. \qedhere
    \end{equation*}
\end{proof}
Given \Cref{prop:g-avg}, to compute $\adv_{\mcD^k}(g\circ f, S) $, it suffices to understand the function $(g \circ f)^{\avg}_{S, \mcD}$. We proceed to transform that function into a form which is easier to understand.
\begin{proposition}
    \label{prop:def-rho-x}
    In the setting of \Cref{lem:lower}, for any $x \in (\bits^n)^k$, let $\nu(x) \in [-1,1]^k$ be the vector where
    \begin{equation*}
        \nu(x)_i \coloneqq \Ex_{\by \sim \mcD^k}[f(\by) \mid x^{(i)}_{S_i} = \by_{S_i}].
    \end{equation*}
    Then,
    \begin{equation*}
        (g \circ f)^{\avg}_{S,\mcD^k}(x) = \Ex_{\bz \sim \pi_{\nu(x)}}[g(\bz)].
    \end{equation*}
\end{proposition}
\begin{proof}
    Consider drawing $\by \sim (\bits^n)^k$ conditioned on $\by_S = x_S$. Let $\bz = f^{\otimes k}(\by)$. By definition,
    \begin{equation*}
        (g \circ f)^{\avg}_{S, \mcD^k}(x) = \Ex[g(\bz)].
    \end{equation*}
    Therefore, we merely need to show that the distribution of $\bz$ is that of $\pi_{\nu(x)}$. For this it is sufficient that,
    \begin{enumerate}
        \item Each $\bz_1, \ldots, \bz_k$ is independent. This follows from the fact $\by_1, \ldots, \by_k$ are independent, and that the restriction that $\by_S = x_S$ is a disjoint restriction for each of the $k$ components.
        \item For each $i \in [k]$, that $\Ex[\bz_i] = \nu(x)_i$. This follows from the definition of $\nu(x)_i$.
    \end{enumerate}
    The desired result follows from the fact that $\pi_{\nu(x)}$ is the unique product distribution over $\bits^k$ with mean $\nu(x)$.
\end{proof}

We now prove the upper bound.
\begin{proof}[Proof of \Cref{lem:lower}]
    Let $\nu$ be as defined in \Cref{prop:def-rho-x}. Applying it and \Cref{prop:g-avg},
    \begin{equation*}
        \adv_{\mcD^k}(g\circ f, S) = \Ex_{\bx \sim \mcD^k}\bracket*{\abs*{\Ex_{\bz \sim \pi_{\nu(\bx)}}[g(\bz)]}} \leq \sqrt{\Ex_{\bx \sim \mcD^k}\bracket*{\paren*{\Ex_{\bz \sim \pi_{\nu(\bx)}}[g(\bz)]}^2}}.
    \end{equation*}
    The inequality above is Jensen's. Consider the random variables $\nu(\bx)_1, \ldots, \nu(\bx)_k$. The have the following two properties.
    \begin{enumerate}
        \item They are independent. This is because the value of $\nu(\bx)_i$ depends on only the value of $\bx_i$, which is independent of the other $\bx_j$ for $j \neq i$.
        \item They each have mean $\mu$. This is because,
        \begin{equation*}
            \Ex[\nu(\bx)_i] = \Ex\bracket*{\Ex_{\by \sim \mcD}[f(\by) \mid (\bx^{(i)})_{S_i} = y_{S_i}]} = \Ex_{\by \sim \mcD}[f(\by)] = \mu.
        \end{equation*}
    \end{enumerate}
    Therefore, we can use \Cref{lem:avg-to-stab}:
    \begin{equation*}
        \Ex_{\bx \sim \mcD^k}\bracket*{\paren*{\Ex_{\bz \sim \pi_{\nu(\bx)}}[g(\bz)]}^2} = \Stab_{\mu, (\Ex[\phi_\mu(\nu(\bx)_1)^2], 
        \ldots, \Ex[\phi_\mu(\nu(\bx)_k)^2])}(g).
    \end{equation*}
    We can further upper bound,
    \begin{align*}
        \Ex[\phi_\mu(\nu(\bx)_i)^2] &\leq \frac{\Ex[\abs{\nu(\bx)_i}] - \mu^2}{1 -  \mu^2} \tag{\Cref{prop:adv-to-square}}\\
        &= \frac{\adv(f, S_i) - \mu^2}{1 -  \mu^2} \tag{\Cref{prop:g-avg}}\\
        &= \beta_i.
    \end{align*}
    Putting the above together, we have that
    \begin{equation*}
        \adv_{\mcD^k}(g\circ f, S) \leq \sqrt{\stab_{\mu, \beta}(g)}. \qedhere
    \end{equation*}
\end{proof}

\subsection{Proofs of the consequences of our strong composition theorem}
\label{subsec:corollary-proofs}
In this section, we complete the proofs of \Cref{cor:error-4} and \Cref{thm:composition theorem intro 1}.
\begin{proof}[Proof of \Cref{thm:composition theorem intro 1}]
    For any partition of the budget junta budget $r_1 + \cdots + r_k = R$, let $\vrho(r_1,\ldots,r_k)$ be the vector,
    \begin{equation*}
        \vrho(r_1,\ldots,r_k)_i \coloneqq \adv_D(f, r_i).
    \end{equation*}
    Then, applying the upper bound on advantage of \Cref{thm:strong-formal} and maximizing over all possible partitions of the budget $R$, we have that
    \begin{equation*}
        \adv_{\mcD^k}(g\circ f, R) \leq \max_{r_1 + \cdots + r_k = R} \sqrt{\stab_{\vrho(r_1, \ldots, r_k)}(g)}.
    \end{equation*}
    This completes the upper bound on the advantage of the optimal $R$-junta approximator of $g \circ f$ of \Cref{thm:composition theorem intro 1}. For the lower bound on the advantage of the optimal \emph{composed form} approximator, let $r_1, \ldots, r_k$ be the partition of budget maximizing $\stab_{\vrho(r_1, \ldots, r_k)}(g)$. Using the lower bound of \Cref{thm:strong-formal}, and using $(\cdot)^2$ to refer to an elementwise squaring of a vector,
    \begin{equation*}
        \adv_{\mcD^k}(g\circ f, h (\tilde{f}_{r_1}, \ldots, \tilde{f}_{r_k})) \geq \stab_{\vrho(r_1,\ldots,r_k)^2}(g).
    \end{equation*}
    Using the Fourier expression for stability \Cref{prop:stab-fourier},
    \begin{align*}
         \Stab_{\vrho(r_1,\ldots,r_k)^2}(g) &= \Ex_{\bS \sim \mcS_{\mu}(g)}\bracket*{((\vrho(r_1,\ldots,r_k)^2)^{\bS}} \\
         &=\Ex_{\bS \sim \mcS_{\mu}(g)}\bracket*{((\vrho(r_1,\ldots,r_k)^{\bS})^{2}} \\
         &\geq\Ex_{\bS \sim \mcS_{\mu}(g)}\bracket*{((\vrho(r_1,\ldots,r_k)^{\bS})}^2 \tag{Jensen's inequality} \\
          &= \Stab_{\vrho(r_1,\ldots,r_k)}(g)^2.
    \end{align*}
    Therefore, there is a composed form approximator with advantage at least $\stab_{\vrho(r_1, \ldots, r_k)}(g)^2$.
\end{proof}

Our proof of \Cref{cor:error-4} uses the following.
\begin{proposition}
    \label{prop:prod-iq}
    For any $\alpha_1,\ldots, \alpha_m \in [0,1]$ and $\beta_1, \ldots, \beta_m \in [0,1]$, satisfying $(1-\alpha_i) \leq 2(1-\beta_i)$ for each $i \in [m]$,
    \begin{equation*}
        1 - \prod_{i \in [m]} \alpha_i \leq 2\paren*{ 1 - \prod_{i \in [m]} \beta_i }.
    \end{equation*}
\end{proposition}
\begin{proof}
We consider the vector $\beta' \in [0,1]^m$ satisfying
\begin{equation*}
    1 - \alpha_i = 2 \cdot (1 - \beta'_i).
\end{equation*}
Note that $\beta'_i \geq \beta_i$, which means that
\begin{equation*}
     1 - \prod_{i \in [m]} \beta'_i  \leq 1 - \prod_{i \in [m]} \beta_i. 
\end{equation*}
Now, consider the function $q:[0,1] \to [0,1]$ defined as 
\begin{equation*}
    q(x) \coloneqq 1 - \prod_{i \in [m]}1 - x(1- \alpha_i).
\end{equation*}
A quick calculation confirms that the second derivative of $q$ is nonpositive, so $q$ is concave. Furthermore, it satisfies,
\begin{align*}
    q(0) &= 0,\\
    q(1) &=  1 - \prod_{i \in [m]} \alpha_i, \\
    q(1/2) &= 1 - \prod_{i \in [m]} \beta'_i.
\end{align*}
We conclude,
\begin{equation*}
    1 - \prod_{i \in [m]} \alpha_i \overset{\text{concavity of $q$}}{\leq} 2\paren*{1 - \prod_{i \in [m]} \beta'_i} \leq \paren*{1 - \prod_{i \in [m]} \beta_i}. \qedhere
\end{equation*}
\end{proof}

\begin{proof}[Proof of \Cref{cor:error-4}]
    Let $r_1 + \cdots + r_k = R$ be the partition of $R$ used in the junta achieving minimum error relative to $g \circ f$ and define, for each $i \in [k]$,
    \begin{align*}
        \alpha_i &\coloneqq \max\paren*{0, \frac{\adv_{\mcD}(f, r_i)^2 - \mu^2}{1 - \mu^2}},\\
        \beta_i &\coloneqq \max\paren*{0, \frac{\adv_{\mcD}(f, r_i) - \mu^2}{1 - \mu^2}},
    \end{align*}
    which satisfy the relation
    \begin{equation*}
        1-\alpha_i \leq 2(1 - \beta_i).
    \end{equation*}
    Applying \Cref{thm:strong-formal} and the relation $\error = \frac{1 - \adv}{2}$, we have that
    \begin{equation*}
        \error_{\mcD^k}(g\circ f, R) \geq \frac{1 - \sqrt{\stab_{\mu, \beta}(g)}}{2},\quad\quad\text{and}\quad\quad\error_{\mcD^k}(g\circ f, h (\tilde{f}_{r_1}, \ldots, \tilde{f}_{r_k})) \leq \frac{1 - \stab_{\mu, \alpha}(g)}{2}.
    \end{equation*}
    Our goal is to show the following series of inequalities, which would imply the desired result,
    \begin{equation*}
        1 - \stab_{\mu, \alpha}(g) \overset{(\text{iq }1)}{\leq} 2(1 - \stab_{\mu, \beta}(g)) \overset{(\text{iq }2)}{\leq} 4(1 - \sqrt{\stab_{\mu, \beta}(g)}).
    \end{equation*}
    The second, (inequality 2), follows the fact that for any $x \in [0,1]$, $(1-x) \leq 2(1-\sqrt{x})$. For the first inequality, using \Cref{prop:stab-fourier}, we can express stability via the Fourier spectrum of $g$ as
    \begin{align*}
        1 - \stab_{\mu, \alpha}(g) &= \sum_{S}\hat{g}(S)^2(1 - \prod_{i \in S}\alpha_i)\\
         &\leq  2\sum_{S}\hat{g}(S)^2(1 - \prod_{i \in S}\beta_i) \tag{\Cref{prop:prod-iq}, $1-\alpha_i \leq 2(1 - \beta_i)$}\\
         & = 2(1 - \stab_{\mu, \beta}(g)).
    \end{align*}
    This proves inequality 1, giving the desired result.
\end{proof}

%% file: Symmetric.tex
\section{Multivariate noise stability of symmetric functions}
In this section, we prove \Cref{lem:intro multivariate to univariate} and \Cref{cor:delta-eps-noise-stability-intro}, connecting the multivariate noise stability of symmetric functions to their univariate noise stability.

\begin{definition}[Symmetric and transitive functions]
    \label{def:sym-trans}
    For any function $g:\bits^k \to \bits$, a permutation $\sigma:[k]\to [k]$ is an automorphism of $g$ if for all inputs $x \in \bits^k$,
    \begin{equation*}
        g(x) = g(x_{\sigma(1)}, \ldots, x_{\sigma(k)}).
    \end{equation*}
    We say $g$ is \emph{symmetric} if every permutation of $[k]$ is an automorphism of $g$. Similarly, $g$ is \emph{transitive} if for all $i,j \in [k]$, there is an automorphism of $g$ sending $i$ to $j$.
\end{definition}

\subsection{The upper bound on the multivariate noise stability of symmetric functions}

\begin{lemma}[Upper bound of \Cref{lem:intro multivariate to univariate}]
    \label{lem:multivariate-upper-sym}
    For any symmetric $g:\bits^k \to \bits$, $\mu \in (-1,1)$, and $\vrho \in [0,1]^k$, let $\rhoup \coloneqq 1/k \cdot \sum_{i \in [k]}\vrho_i$. Then,
    \begin{equation*}
        \stab_{\mu, \vrho}(g)\leq \stab_{\mu, \rhoup}(g).
    \end{equation*}
\end{lemma}

Our proof of \Cref{lem:multivariate-upper-sym} will use make heavy use of the \emph{negative association} of random variables.
\begin{definition}[Negative association \cite{JP83}]
    \label{def:NA}
    A set of random variables $\bx_1, \ldots, \bx_m$ supported on $\R$ are \emph{negatively associated} if for all \emph{disjoint} subsets $S_1, S_2 \subseteq [m]$ and $S_1$-juntas $f_1:\R^m \to \R$, $S_2$-juntas $f_2:\R^m \to \R$ both monotonically nondecreasing,
    \begin{equation*}
        \Ex[f_1(\bx)f_2(\bx)] \leq \Ex[f_1(\bx)]\Ex[f_2(\bx)].
    \end{equation*}
\end{definition}
For our purposes, we will only need a few useful facts about negatively associated random variables given in \cite{JP83} (see also \cite{waj17} for a useful overview).
\begin{fact}[Permutation distributions are negatively associated, \cite{JP83}]
    \label{fact:NA-perm}
    For any $z_1, \ldots, z_m \in \R$, draw a uniformly random permutation $\bsigma:[m] \to [m]$ and set $\bx_i = z_{\bsigma(i)}$ for each $i \in [k]$. Then, $\bx_1, \ldots, \bx_m$ are negatively associated.
\end{fact}
\begin{fact}[Subsets of negatively associated random variables are negatively associated]
    \label{fact:NA-subset}
    For any $2 \leq m' \leq m$, if $\bx_1, \ldots, \bx_m$ are negatively associated, then $\bx_1, \ldots, \bx_{m'}$ are also negatively associated.
\end{fact}
\begin{fact}[Product consequence of negative association]
    \label{fact:NA-prod}
    For any negatively associated $\bx_1, \ldots, \bx_m$ and nondecreasing $f:\R \to \R_{\geq 0}$,
    \begin{equation*}
        \Ex\bracket*{\prod_{i \in [m]}f(\bx_i)} \leq \prod_{i \in [m]}\Ex\bracket*{f(\bx_i)}.
    \end{equation*}
\end{fact}

Given the above, facts about negative associated random variables, we can now prove \Cref{lem:multivariate-upper-sym}.
\begin{proof}[Proof of \Cref{lem:multivariate-upper-sym}]
    We expand $\stab_{\mu, \vrho}(g)$ using the Fourier spectrum of $g$ (\Cref{prop:stab-fourier}),
    \begin{equation*}
        \Stab_{\mu, \vrho}(g) = \Ex_{\bS \sim \mcS_{\mu}(g)}[(\vrho)^{\bS}].
    \end{equation*}
    Let $\bell$ be the distributed the same as $|\bS|$ for $\bS \sim \mcS_{\mu}(g)$. Then,
    \begin{equation*}
        \Stab_{\mu, \vrho}(g) = \Ex_{\bell}\bracket*{\Ex_{\bS \sim \mcS_{\mu}(g)}[(\vrho)^{\bS} \mid |\bS| = \ell]}.
    \end{equation*}
    Since $g$ is symmetric, for any $|S_1| = |S_2|$, $\hat{g}(S_1) = \hat{g}(S_2)$. As a result the distribution of $\bS \sim \mcS_{\mu}(g)$ conditioned on $|\bS| = \ell$ is simply a uniformly random size-$\ell$ subset of $[k]$. Formally,
    \begin{equation*}
         \Stab_{\mu, \vrho}(g) = \Ex_{\bell}\bracket*{\Ex_{\bS \sim \binom{[k]}{\bell}}[(\vrho)^{\bS}]}.
    \end{equation*}
    Let $\bx_1, \ldots, \bx_k$ be a uniform random permutation of $\vrho_1, \ldots, \vrho_k$. Then, the distribution of $(\vrho)^{\bS}$ for $\bS \sim \binom{[k]}{\ell}$ is identical to that of $\prod_{i \in [\ell]}\bx_i$. By \Cref{fact:NA-perm,fact:NA-subset}, $\bx_1, \ldots, \bx_\ell$ are negatively associated, and so,
    \begin{equation*}
        \Ex_{\bS \sim \binom{[k]}{\ell}}[(\vrho)^{\bS}] = \Ex\bracket*{\prod_{i \in [\ell]}\bx_i} \overset{\mathrm{(\Cref{fact:NA-prod})}}{\leq} \prod_{i \in [\ell]}\Ex[\bx_i] = \paren*{\rhoup}^\ell.
    \end{equation*}
    Therefore,
    \begin{equation*}
         \Stab_{\mu, \vrho}(g) \leq \Ex_{\bell}\bracket*{\paren*{\rhoup}^{\bell} } = \stab_{\mu, \rhoup}(g).\qedhere
    \end{equation*}
\end{proof}

\subsection{The lower bound on the multivariate noise stability of symmetric functions}
\begin{lemma}[Lower bound of \Cref{lem:intro multivariate to univariate}]
    \label{lem:multivariate-lower-sym}
    For any transitive $g:\bits^k \to \bits$, $\mu \in (-1,1)$, and $\vrho \in [0,1]^k$, let $\rholow \coloneqq \paren*{\prod_{i \in [k]}\vec{\rho}_i}^{1/k}$. Then,
    \begin{equation*}
        \stab_{\mu, \vrho}(g)\geq \stab_{\mu, \rholow}(g).
    \end{equation*}
\end{lemma}
Note that every transitive $g$ is also symmetric, but the reverse does not hold.
\begin{proof}
    Similarly to the proof of \Cref{lem:multivariate-upper-sym}, let $\bell$ be the distribution of $|\bS|$ when $\bS \sim \mcS_{\mu}(g)$. Then,
    \begin{equation*}
        \Stab_{\mu, \vrho}(g) = \Ex_{\bell}\bracket*{\Ex_{\bS \sim \mcS_{\mu}(g)}[(\vrho)^{\bS} \mid |\bS| = \bell]}.
    \end{equation*}
    For each $S \subseteq [k]$, we'll use $\chi(S) \in \zo^k$ to denote the characteristic vector of $S$, meaning $\chi(S)_i \coloneqq \Ind[i \in S]$. Then,
    \begin{align*}
        \Stab_{\mu, \vrho}(g) &=  \Ex_{\bell}\bracket*{\Ex_{\bS \sim \mcS_{\mu}(g)}\bracket*{\prod_{i \in [k]} (\vrho_i)^{\chi(\bS)_i} \,\bigg|\, |\bS| = \bell}}\\
        &=  \Ex_{\bell}\bracket*{\Ex_{\bS \sim \mcS_{\mu}(g)}\bracket*{\exp\paren*{\sum_{i \in [k]} \chi(\bS)_i \log(\vrho_i)}\,\bigg|\,|\bS| = \bell}}\\
        &\geq  \Ex_{\bell}\bracket*{\exp\paren*{\Ex_{\bS \sim \mcS_{\mu}(g)}\bracket*{\sum_{i \in [k]} \chi(\bS)_i \log(\vrho_i)\,\bigg|\,|\bS| = \bell}}} \tag{Jensen's inequality}\\
        &=  \Ex_{\bell}\bracket*{\exp\paren*{\sum_{i \in [k]} \log(\vrho_i) \Prx_{\bS \sim \mcS_{\mu}(g)}\bracket*{i \in \bS \mid |\bS| = \bell}}}. \tag{Linearity of expectation}
    \end{align*}
    Fix any $i_1, i_2 \in [k]$ and level $\ell \in [0,k]$. Since $g$ is transitive, there is an automorphism, $\sigma$, of $g$ sending $i_1$ to $i_2$. Since $\sigma$ is an automorphism of $g$, for any $S \subseteq [k]$, for $\bS \sim \mcS_{\mu}(g)$, $\Pr[\bS = S] = \Pr[\bS = \sigma(S)]$. As a result
    \begin{equation*}
        \Prx_{\bS \sim \mcS_{\mu}(g)}\bracket*{i_1 \in \bS \mid |\bS| = \ell} = \Prx_{\bS \sim \mcS_{\mu}(g)}\bracket*{i_2 \in \bS \mid |\bS| = \ell},
    \end{equation*}
    and so $\Prx_{\bS \sim \mcS_{\mu}(g)}\bracket*{i \in \bS \mid |\bS| = \ell}$ must be the same for all $i \in [k]$. The sum of these probabilities is $\ell$, meaning each is $\frac{\ell}{k}$. This allows us to bound,
    \begin{align*}
        \Stab_{\mu, \vrho}(g) &\geq \Ex_{\bell}\bracket*{\exp\paren*{\sum_{i \in [k]} \log(\vrho_i) \cdot \frac{\bell}{k}}} \\
        &=\Ex_{\bell}\bracket*{\prod_{i \in [k]}\paren*{\vrho_i}^{\frac{\bell}{k}}} \\
        &=\Ex_{\bell}\bracket*{(\rholow)^{\bell}} =  \stab_{\mu, \rholow}(g).\qedhere
    \end{align*}
\end{proof}
\subsection{Bounding the $(\delta,\eps)$-noise stability of symmetric functions}
Recall, from \Cref{def:delta-eps-NS}, that the $(\delta,\eps)$-noise stability of a function $g:\bits^k\to\bits$ is the quantity
\[ \max\big\{ \Stab_{\vec{\rho}}(g)\colon  \text{at least $\delta$-fraction of $\vec{\rho}$'s coordinates are at most $1-2\eps$}\big\}.\] 
We prove \Cref{cor:delta-eps-noise-stability-intro}, restated below.

\begin{corollary}[Formal version of \Cref{cor:delta-eps-noise-stability-intro}]
    \label{cor:AM-GM-close}
    For any symmetric function $g:\bits^k \to \bits$, $\delta \in (0,1)$, and $\eps \in (0,1/2)$, let $\delta'\coloneqq \frac{\ceil{k\delta}}{k}$ be $\delta$ rounded up to the nearest integer multiple of $1/k$. Then, the $(\delta, \eps)$-noise stability of $g$ is equal to $\stab_{\mu, \rho^\star}(g)$ for some $\rho^\star$ satisfying
    \begin{equation*}
        1 - 2\delta'\eps - 4\eps^2 \leq \rho^\star \leq 1 - 2\delta'\eps.
    \end{equation*}
\end{corollary}
\begin{proof}[Proof of \Cref{cor:AM-GM-close}]
    Since stability is monotone (\Cref{cor:multi-monotone}), the $(\delta, \eps)$-noise stability of $g$ is its multivariate noise stability with a correlation vector $\vrho$ where $\delta'$ fraction of the coordinates are $1 - 2\eps$ and the remainder are $1$. The arithmetic mean of this vector is exactly $1 - 2\delta'\eps$, and its geometric mean is $(1 - 2\eps)^{\delta'}$. The desired result then follows from \Cref{lem:multivariate-lower-sym,lem:multivariate-upper-sym} and the inequality
    \begin{equation*}
        (1 - x)^c \geq 1-cx - (1-c)x^2 \geq 1 - cx - x^2
    \end{equation*}
    which holds for all $c,x \in [0,1]$. To prove this inequality, it is sufficient that $q_c(x) \geq 0$ for all $x,c \in [0,1]$ where
    \begin{equation*}
        q_c(x) \coloneqq (1-x)^c - 1 +cx + (1-c)x^2.
    \end{equation*}
To see this, we note that for any $c \in [0,1]$, the function $q_c(x)$ has roots at $x = 0$ and $x=1$. It is furthermore increasing at $x = 0$, and decreasing at $x = 1$. If $q_c(x)$ were to be negative for any $x \in [0,1]$, then, it would need to have at least $3$ local extrema. However, the derivative $q_c'(x)$ is concave, so it can only be zero at a maximum of $2$ points. This proves the desired inequality. (If the reader prefers, \Cref{fig:q-proof-by-pic} gives a ``proof by picture".) 
\end{proof}

\begin{figure}[htb]
\centering

    \begin{tikzpicture}
    \begin{axis}[
        xlabel=$x$,
        ylabel style={rotate=-90},
        ylabel=$q_c(x)$,
        ymin=0,
        xmin=0,
        xmax=1,
        width=10cm,
        height=6cm,
        legend pos=north west,
        legend style={font=\footnotesize},
        legend cell align=left,
    ]
    
    \addplot[
        domain=0:1,
        samples=100,
        color=blue!80!black,
        densely dashed,
        ultra thick,
        ]{(1-x)^0.3 - 1 + 0.3*x + 0.7*x^2};
    \addlegendentry{$c=0.3$}
    
    \addplot[
        domain=0:1,
        samples=100,
        color=red!80!black,
        dotted,
        ultra thick,
        ]{(1-x)^0.5 - 1 + 0.5*x + 0.5*x^2};
    \addlegendentry{$c=0.5$}
    
    \addplot[
        domain=0:1,
        samples=100,
        color=green!70!black,
        densely dotted,
        ultra thick,
        ]{(1-x)^0.7 - 1 + 0.7*x + 0.3*x^2};
    \addlegendentry{$c=0.7$}
    
    \end{axis}
    \end{tikzpicture}
    \caption{Plots of $q_c$ defined in the proof of \Cref{cor:AM-GM-close} for various values of $c$, showing that $q_c(x)\ge 0$ for all $x \in [0,1]$.}
    \label{fig:q-proof-by-pic}

\end{figure}
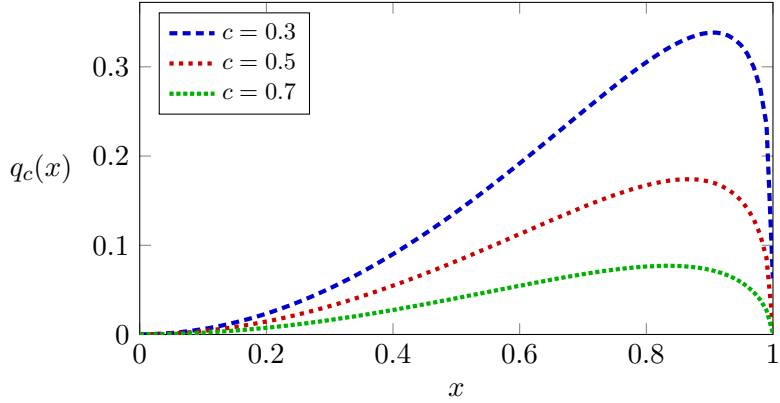

%% file: property_testing.tex
\section{Composition theorems yield boosters for property testing}

\subsection{A general boosting framework}
Let $\mathcal{P}=\left\{\mathcal{P}_s\right\}_{s\in \N}$ be a parametrized property of Boolean functions. For a function $f:\bits^n\to\bits$ and distribution $\mathcal{D}$ over $\bits^n$, we write
$$
\dist_{\mathcal{D}}(f,\mathcal{P}_s)\coloneqq \min_{h\in \mathcal{P}_s}\error_{\mathcal{D}}(f,h)
$$
to denote $f$'s distance to $\mathcal{P}_s$ over $\mathcal{D}$. We are interested in the \textit{relaxed} testing regime for size parameters $s>s'$ where we want to decide whether an unknown target function $f$ belongs to $\mathcal{P}_s$ or is $\eps$-far from $\mathcal{P}_{s'}$ under $\mathcal{D}$: $\dist_{\mathcal{D}}(f,\mathcal{P}_{s'})>\eps$ (recall \Cref{def:testing of P under D}). We say that $\mathcal{P}$ is $(\eps,s,s')$-testable if there exists an algorithm for $(\eps,s,s')$-testing $\mathcal{P}$ for every distribution $\mathcal{D}$. As $\eps\to 0$, the gap between the Yes and No cases becomes smaller and $(\eps,s,s')$-testing becomes more difficult. The main result of this section is that if $\mathcal{P}$ ``behaves well'' under function composition, then testers for large $\eps$ can be boosted to testers for the more challenging regime of small $\eps$. We will specialize our attention to properties which behave \textit{linearly} with respect to function composition.

\begin{definition}[Linear with respect to composition]
\label{def:linear wrt composition}
    A parametrized property $\mathcal{P}=\left\{\mathcal{P}_s\right\}_{s\in \N}$ behaves linearly (with respect to function composition) if 
    $$
    f\in\mathcal{P}_s\quad \Rightarrow\quad g\circ f\in\mathcal{P}_{k\cdot s}
    $$
    for all $g:\bits^k\to\bits$, $f:\bits^n\to\bits$, and $s\in\N$.
\end{definition}

\paragraph{Examples.}{
Being an $s$-junta, depth-$s$ decision tree, depth-$s$ formula, or degree-$s$ polynomial are all properties of Boolean functions which behave linearly with respect to composition. As is often the case, it is straightforward to show from their definitions that these properties behave linearly. Many properties which do not \textit{a priori} behave linearly can be converted into ones that do by applying an appropriate transformation to their size. For example, the property $\mathcal{P}_s=\{\text{size-exp}(s)\text{ decision trees}\}$ behaves linearly. 
}

\paragraph{Strong composition theorems for properties.}{
A property $\mathcal{P}$ which behaves linearly with respect to function composition is said to admit a \textit{strong composition theorem} if the upper bound from \Cref{def:linear wrt composition} can be shown to be nearly tight. This definition generalizes the relation \ref{eq:strong amplification intro}. 

\begin{definition}[$(\smalle,\largee,\lambda)$-composition theorem]
\label{def:strong composition theorem}
    A parametrized property $\mathcal{P}=\left\{\mathcal{P}_s\right\}_{s\in \N}$ admits an $(\smalle,\largee,\lambda)$-composition theorem with respect to $g:\bits^k\to \bits$ for $\smalle,\largee\in (0,1)$ and a constant $\lambda>0$ if
    $$
    \dist_{\mathcal{D}}(f,\mathcal{P}_s)>\smalle\quad\Rightarrow\quad \dist_{\mathcal{D}^k}(g\circ f,\mathcal{P}_{\lambda ks})>\largee
    $$
    for all $f:\bits^n\to\bits$ and distributions $\mathcal{D}$ over $\bits^n$.
\end{definition}
Strong composition theorems depend on the combining function $g$. For example, if $g$ is a constant function then one would not expect the upper bound from \Cref{def:linear wrt composition} to be tight. For this reason, the dependence on $g$ is made explicit in the definition of strong composition theorem. 
 Roughly speaking, the definition says that if a property $\mathcal{P}$ behaves linearly and admits a strong composition theorem with respect to $g$, then composing with $g$ turns a function in $\mathcal{P}_s$ into one in $\mathcal{P}_{s k}$ and turns a function slightly far from $\mathcal{P}_{s}$ into one very far from $\mathcal{P}_{ \Theta(s k)}$. For a fixed $\largee$, having an $(\smalle,\largee,\lambda)$-composition theorem with respect to $g$ becomes stronger as $\smalle$ approaches $0$. In general, we are interested in $(\smalle,\largee,\lambda)$-composition theorems when $\largee\gg \smalle$. The parameter $\lambda$ is built into the definition to tolerate a small amount of slack between the upper and lower bounds on $g\circ f$. For many applications, this constant factor is necessary. We are now equipped to state our main boosting theorem.
}




    

\begin{theorem}[Boosting property testers, formal version of \Cref{thm:generic boost intro}]
\label{thm:generic boosting}
Let $\mathcal{P}=\{\mathcal{P}_s\}_{s\in\N}$ be a property which behaves linearly and admits an $(\smalle,\largee,\lambda)$-composition theorem with respect to $g:\bits^k\to\bits$. If $\mathcal{P}$ is $(\largee,s,s')$-testable in $q(\largee,s,s')$ queries, then it is $(\smalle,s,\lambda^{-1} s')$-testable using $k\cdot q(\largee,ks,ks')$ many queries. 
\end{theorem}

\begin{proof}
Let $\weak$ be an algorithm for $(\largee,s,s')$-testing $\mathcal{P}$. Given queries to a function $f:\bits^{n}\to\bits$ and random samples from a distribution $\mathcal{D}$ over $\bits^n$, we $(\smalle,s, \lambda^{-1} s')$-test $\mathcal{P}$ using the procedure in \Cref{fig:boosting testers} where $\weak$ is given an instance of $(\largee,ks,ks')$-testing $\mathcal{P}$. 

\begin{figure}[h!]
\begin{tcolorbox}[colback = white,arc=1mm, boxrule=0.25mm]
\vspace{3pt}
$\strong$, a boosted tester for $\mathcal{P}$:
\begin{itemize}[leftmargin=10pt,align=left]
\item[\textbf{Given:}] $\weak$, a weak tester for $\mathcal{P}$; queries to $f$; and random samples from $\mathcal{D}$
\item[\textbf{Run:}]Simulate $\weak$ providing it with
\begin{itemize}[align=left,labelsep*=0pt]
    \item \textit{queries}: return $g(f(x^{(1)}),\ldots,f(x^{(k)}))$ for a query $(x^{(1)},\ldots,x^{(k)})\in(\bits^{n})^k$; and
    \item \textit{random samples}: return $(\bx^{(1)},\ldots,\bx^{(k)})\sim \mathcal{D}^k$ from $k$ independent samples $\bx^{(i)}\sim \mathcal{D}$.
\end{itemize}
\item[\textbf{Output:}] Yes if and only if $\weak$ outputs Yes
\end{itemize}
\vspace{3pt}
\end{tcolorbox}
\medskip
\caption{Boosting a weak tester via function composition.}
\label{fig:boosting testers}
\end{figure}

\paragraph{Query complexity.}{
The target $g\circ f:\bits^{nk}\to\bits$ is a $(\largee, ks,ks')$-testing instance for $\weak$. Therefore, $\weak$ makes $q(\largee,ks,ks')$ queries to the target $g\circ f:\bits^{nk}\to\bits$ before terminating. Our tester makes $k$ queries to $f$ for each query to $g\circ f$. So our tester for $f$ makes $k\cdot q(\largee,ks,ks')$ queries in total.
}
\paragraph{Correctness.}{In the Yes case, $f\in \mathcal{P}_{s}$. We then have $g\circ f\in \mathcal{P}_{sk}$ since $\mathcal{P}$ behaves linearly. This ensures that $\weak$ outputs Yes.  In the No case, $\dist_{\mathcal{D}}(f,\mathcal{P}_{s'/\lambda})>\smalle$. We then have $\dist_{\mcD^k}(g\circ f,\mathcal{P}_{ks'})>\lambda$ since $\mathcal{P}$ admits an $(\largee,\smalle,\lambda)$-composition theorem. This ensures that $\weak$ outputs No.}
\end{proof}

\subsection{Implications for current landscape of junta testing} 

Our results have new implications for \textit{tolerantly} testing juntas. In this regime, the Yes case of \Cref{def:testing of P under D} is relaxed to only require that $f$ is \textit{close} to an $r$-junta over $\mathcal{D}$. 

\begin{definition}[Tolerantly $(\yese,\noe,r,r')$-testing juntas]
\label{def:tolerant junta testing}
    Given parameters $r\le r'$ and $\yese\le \noe$, queries to an unknown function $f:\bits^n\to\bits$, and random samples from a distribution $\mathcal{D}$ over $\bits^n$, distinguish between
    \begin{itemize}
        \item Yes: $f$ is $\yese$-close to being an $r$-junta under $\mathcal{D}$, and
        \item No: $f$ is $\noe$-far from being an $r'$-junta under $\mathcal{D}$.
    \end{itemize}
\end{definition}

In all of our applications, we will be using \Cref{thm:generic boosting}, or a variant of it, with $g$ set to $\XOR_k$. For this reason, we start with some useful properties about the noise stability of parity.

\subsubsection{Noise stability of parity under general product distributions}
\begin{lemma}
    \label{lem:par-boost-error}
    For any $f:\bits^n \to \bits$, distribution $\mcD$ over $\bits^n$, junta budget $R$, and $R$-junta $h$,
    \begin{equation*}
        \error_{\mcD^k}(\XOR_k \circ f, h) \geq \min_{r_1+\cdots+r_k=R}\frac{1 - \sqrt{\prod_{i \in [k]}\paren*{1 - 2\cdot\error_{\mcD}(f, \tilde{f}_{r_i})}}}{2}.
    \end{equation*}
\end{lemma}

Our proof of \Cref{lem:par-boost-error} will use the multivariate noise stability of parity.
\begin{proposition}[The multivariate noise stability of parity]
    \label{prop:multi-parity}
    For any $\mu \in (-1,1)$, $\vec{\rho} \in [0,1]^k$,
    \begin{equation*}
        \stab_{\mu, \vrho}(\XOR_k) = \prod_{i \in [k]}\paren*{\vrho_i + (1-\vrho_i)\cdot \mu^2}=\prod_{i \in [k]}\paren*{1 - (1-\vrho_i)(1-\mu^2)}.
    \end{equation*}
\end{proposition}
\begin{proof}
    Note that $\XOR_k(y_1, \ldots, y_k) = \prod_{i \in [k]}y_i$. Therefore,
    \begin{equation*}
        \stab_{\mu, \vrho}(\XOR_k) = \Ex_{\by \sim (\pi_{\mu})^k, \bz \rhosim \by}\bracket*{\prod_{i \in [k]} \by_i \bz_i}.
    \end{equation*}
    Each pair $(\by_i, \bz_i)$ are independent of another, so
    \begin{equation*}
        \stab_{\mu, \vrho}(\XOR_k) = \prod_{i \in [k]} \Ex\bracket*{\by_i \bz_i}.
    \end{equation*}
    The distribution of $(\by_i, \bz_i)$ can be succinctly described: With probability $\vrho_i$, $\bz_i = \by_i$. Otherwise, they are each independent draws from $\pi_{\mu}$. Therefore,
    \begin{equation*}
         \Ex\bracket*{\by_i \bz_i} = \vrho_i + (1-\vrho_i)\cdot \mu^2.
    \end{equation*}
    The desired result follows from combining the above equations
\end{proof}

\begin{proof}[Proof of \Cref{lem:par-boost-error}]
    We apply our strong composition theorem, \Cref{thm:strong-formal}. It is stated in terms of advantage and gives
    \begin{equation*}
        \max_{R\text{-juntas }h} \adv_{\mcD^k}(\XOR_k \circ f, h) \leq \max_{r_1 + \cdots + r_k = R}\sqrt{\stab_{\mu, \beta(r_1, \ldots, r_k)}(\XOR_k)},
    \end{equation*}
    where we define $\mu = \Ex_{\bx \sim \mcD}[f(\bx)]$, and $\beta(r_1, \ldots, r_k) \in [0,1]^k$ is the vector
    \begin{equation*}
        \beta(r_1, \ldots, r_k)_i = \frac{\adv_{\mcD}(f, \tilde{f}_{r_i}) - \mu^2}{1 - \mu^2} = \frac{1 - 2\cdot\error_{\mcD}(f, \tilde{f}_{r_i}) - \mu^2}{1 - \mu^2}.
    \end{equation*}
    Applying \Cref{prop:multi-parity},
    \begin{align*}
        \max_{R\text{-juntas }h} \adv_{\mcD^k}(\XOR_k \circ f, h) &\leq \max_{r_1 + \cdots + r_k = R}\sqrt{\prod_{i \in [k]}\paren*{1 - \paren*{1 -\frac{1 - 2\cdot\error_{\mcD}(f, \tilde{f}_{r_i}) - \mu^2}{1 - \mu^2} }(1 - \mu^2)}}\\
        &= \max_{r_1 + \cdots + r_k = R}\sqrt{\prod_{i \in [k]}\paren*{1 - \paren*{\frac{2\cdot\error_{\mcD}(f, \tilde{f}_{r_i})}{1 - \mu^2} }(1 - \mu^2)}}\\
        &= \max_{r_1 + \cdots + r_k = R}\sqrt{\prod_{i \in [k]}\paren*{1 - 2\cdot\error_{\mcD}(f, \tilde{f}_{r_i})}}.
    \end{align*}
    The desired result follows from $\error = \frac{1 - \adv}{2}$.
\end{proof}

\subsubsection{Warmup: weak testers suffice for $(0,\eps,r,r')$-testing juntas}

We first boost tolerant testers in the regime where $\yese$ is fixed to $0$ in \Cref{def:tolerant junta testing}. This version is slightly easier to state and is also the version we will use later in proving \Cref{thm:np hardness of tolerant junta testing}.

\begin{theorem}[Boosting $(0,\eps,r,r')$-testers for juntas]
\label{thm:boosting 0 error juntas}
    If juntas can be $(0,\largee,r,r')$-tested using $q(\largee,r,r')$ queries, then for all $k\in \N$ and $\lambda\in (0,1)$, they can be $(0,\smalle,r,\lambda^{-1} r')$-tested in $k\cdot q(\largee,kr,kr')$ queries where
    $$
    \largee=\frac{1-(1-2\smalle)^{{(1-\lambda)k}/{2}}}{2}.
    $$
\end{theorem}

We will need to following composition theorem for juntas. It is a more precise version of \Cref{cor:intro amplify junta complexity} stated in terms of \Cref{def:strong composition theorem}.
\begin{lemma}
\label{lem:junta composition theorem for xor}
    For any $\lambda\in (0,1)$, the property of being an $r$-junta admits an $(\smalle, \largee,\lambda)$-composition theorem with respect to $\XOR_k$ for any $\smalle\le \largee$ where
    $$
    \largee= \frac{1-(1-2\smalle)^{{(1-\lambda)k}/{2}}}{2}.
    $$
\end{lemma}

\begin{proof}
    Assume that $f:\bits^n\to\bits$ is $\smalle$-far from being an $r$-junta over $\mathcal{D}$. We would like to show that $\XOR_k\circ f$ is $\largee$-far from being a $\lambda r k$-junta over $\mathcal{D}^k$ where $\largee$ is defined as in the lemma statement. Let $r_1+\cdots+r_k=\lambda rk$ be the partition of the junta budgets which minimizes the expression 
    $$
    \frac{1 - \sqrt{\prod_{i \in [k]}\paren*{1 - 2\cdot\error_{\mcD}(f, \tilde{f}_{r_i})}}}{2}
    $$
    from \Cref{lem:par-boost-error}. Let $A_{\le r}\sse [k]$ denote the indices for which $r_i\le r$ and let $A_{>r}=[k]\setminus A_{\le r}$.  By a counting argument, at least a $(1-\lambda)$-fraction of $r_i$ satisfy $r_i\le r$ and so $|A_{\le r}|\ge (1-\lambda)k$. By our assumption that $f$ is far from being an $r$-junta, for these $r_i$, we get $\error_\mathcal{D}(f,\Tilde{f}_{r_i})>\smalle$. Therefore, we can conclude that for any $\lambda rk$-junta $h:\bits^{nk}\to\bits$:
    \begin{align*}
        \error_{\mathcal{D}^k}(\XOR_k\circ f,h)&\ge  \frac{1 - \sqrt{\prod_{i \in [k]}\paren*{1 - 2\cdot\error_{\mcD}(f, \tilde{f}_{r_i})}}}{2}\tag{\Cref{lem:par-boost-error}}\\
        &=\frac{1 - \sqrt{\prod_{i \in A_{\le r}}\paren*{1 - 2\cdot\error_{\mcD}(f, \tilde{f}_{r_i})}\cdot\prod_{i\in A_{>r}}\paren*{1 - 2\cdot\error_{\mcD}(f, \tilde{f}_{r_i})}}}{2}\\
        &\ge  \frac{1 - \sqrt{\prod_{i \in A_{\le r}}\paren*{1 - 2\cdot\error_{\mcD}(f, \tilde{f}_{r_i})}}}{2}\tag{$\error\le \frac{1}{2}$}\\
        &> \frac{1 - \paren*{1 - 2\smalle}^{(1-\lambda)k/2}}{2}\tag{$\error_\mathcal{D}(f,\Tilde{f}_{r_i})>\smalle$ for $i\in A_{\le r}$}.
    \end{align*}
    Since $h$ was arbitrary, this shows that $\XOR_k\circ f$ is $\largee$-far from being a $\lambda rk$-junta. 
\end{proof}

\begin{proof}[Proof of \Cref{thm:boosting 0 error juntas}]
\Cref{thm:generic boosting} is stated in the non-tolerant regime. However, we note that the same theorem holds in the $(0,\eps,r,r')$-testing regime. That is, under the conditions of \Cref{thm:generic boosting}, if $\mathcal{P}$ is $(0,\largee,s,s')$-testable, then it is also $(0,\smalle,s,\lambda^{-1}s')$-testable. This is because if $\Tilde{f}$ is a $0$-approximator of $f$ over $\mathcal{D}$, then $g\circ \Tilde{f}$ is a $0$-approximator of $g\circ f$ over $\mathcal{D}^k$. 

\Cref{lem:junta composition theorem for xor} shows that the property of being an $r$-junta admits an $(\smalle, \frac{1-(1-2\smalle)^{{(1-\lambda)k}/{2}}}{2},\lambda)$-composition theorem. Therefore, \Cref{thm:generic boosting} shows that if juntas can be $(0,\largee,r,r')$-tested in $q(\largee,r,r')$ queries then they can be $(\smalle, r,r')$-tested in $k\cdot q(\largee,kr,kr')$ queries where
\[
\largee=\frac{1-(1-2\smalle)^{{(1-\lambda)k}/{2}}}{2}.\qedhere
\]
\end{proof}

\subsubsection{Weak testers suffice for tolerant junta testing}

\begin{theorem}[Boosting tolerant junta testers, formal version of \Cref{cor:boost juntas intro}]
    \label{thm:boosting tolerant junta testers}
    If there is a $q(r)$-query tester that, given queries to $f:\bits^n\to\bits$ and random samples from a distribution $\mathcal{D}$, distinguishes between
    \begin{itemize}
        \item Yes: $f$ is $\frac{1}{4}$-close to an $r$-junta, and
        \item No: $f$ is $\frac{1}{3}$-far from every $r$-junta,
    \end{itemize}
    then for every $\eps>0$ and $\lambda\in (0,1)$, there is a $\frac{q(r/(4\eps))}{4\eps}$-query algorithm that distinguishes between
    \begin{itemize}
        \item Yes: $f$ is $\eps$-close to an $r$-junta, and
        \item No: $f$ is $\Omega(\frac{\eps}{1-\lambda})$-far from every $\lambda^{-1}r$-junta.
    \end{itemize}
\end{theorem}

\begin{proof}
    Let $\mathcal{T}$ be a $q(r)$-query tester for juntas that satisfies the theorem statement. Given queries to a function $f:\bits^n\to\bits$ and random samples to $\mcD$, we design an algorithm for $(\eps,\frac{5}{1-\lambda}\eps, r,\lambda^{-1}r)$-testing $f$ over $\mcD$. The algorithm is straightforward. We choose $k=\frac{1}{4\eps}$, and run the procedure in \Cref{fig:boosting testers} with $g=\XOR_k:\bits^{k}\to\bits$ and junta size $kr$. 
    \paragraph{Query complexity.}{$\mathcal{T}$ makes $q(kr)=q(\frac{r}{4\eps})$ queries to the target $\XOR_k\circ f:\bits^{nk}\to\bits$ before it terminates. Our tester makes $k$ queries to $f$ for each  query to $\XOR_k\circ f$. Therefore, our tester makes $k\cdot q(\frac{r}{4\eps})=(\frac{r}{4\eps})/(4\eps)$ queries in total. 
    }
    \paragraph{Correctness.}{For correctness, we need to show:
    \begin{itemize}[leftmargin=40pt,align=left]
        \item[\underline{Yes case}:] if $f$ is $\eps$-close to being an $r$-junta over $\mcD$, then $\XOR_k\circ f$ is $1/4$-close to being a $kr$-junta over $\mcD^k$, and
        \item[\underline{No case}:] if $f$ is $\frac{5\eps}{1-\lambda}$-far from being an $\lambda^{-1}r$-junta over $\mcD$, then $\XOR_k\circ f$ is $\frac{1}{3}$-far from being a $kr$-junta over $\mcD^k$.
    \end{itemize}

    \subparagraph{Yes case.}{
        Let $\Tilde{f}$  be an $r$-junta which $\eps$-approximates $f$ over $\mcD$. By a union bound:
        \begin{align*}
            \Prx_{\bx\sim \mathcal{D}^k}\left[\text{XOR}_k\circ f(\bx)\neq \text{XOR}_k\circ \Tilde{f}(\bx)\right]&\le \Prx_{\bx\sim\mathcal{D}^k}\left[\text{some }f(\bx^{(i)})\neq f(\bx^{(i)})\right]\\
            &\le k\cdot\error_{\mathcal{D}}(f,\Tilde{f})\le k\eps = \frac{1}{4}.
        \end{align*}
        Since $\XOR_k\circ \Tilde{f}$ is a $kr$-junta, this shows that $\XOR_k\circ f$ is $\frac{1}{4}$-close to a $kr$-junta.
    }
    \subparagraph{No case.}{
        If $f$ is $\frac{5\eps}{(1-\lambda)}$-far from being a $\lambda^{-1}r$-junta, then \Cref{lem:junta composition theorem for xor} implies that $\XOR_k\circ f$ is
        $$
        \frac{1-(1-2\smalle)^{(1-\lambda)k/2}}{2}
        $$
        far from being a $\lambda \lambda^{-1}kr=kr$-junta over $\mcD^k$ where $\smalle\coloneqq \frac{5\eps}{(1-\lambda)}$. Therefore, it is sufficient to show that $\frac{1-(1-2\smalle)^{(1-\lambda)k/2}}{2}\ge \frac{1}{3}$. We observe $\frac{2}{(1-\lambda)k}\le \log_{\frac{1}{3}}(e)\cdot \smalle$ which implies $3^{-2/((1-\lambda)k)}\ge e^{-2\smalle}\ge 1-2\smalle$. It follows:
        \begin{equation*}
            \frac{1}{3}\ge (1-2\smalle)^{(1-\lambda)k/2}
        \end{equation*}
        which provides the desired bound.
    }}
\end{proof}

\subsubsection{Hardness of distribution-free tolerant junta testing} 


We prove the following which implies \Cref{cor:NP hardness intro}.
\begin{theorem}[Tolerant junta testing hardness, formal version of \Cref{cor:NP hardness intro}]
\label{thm:np hardness of tolerant junta testing}
    Given queries to a function $f:\bits^n\to\bits$ and random samples from a distribution $\mathcal{D}$, and $r\le n$, it is \textnormal{NP}-hard under randomized reductions to distinguish between
    \begin{itemize}
        \item Yes: $f$ is $0$-close an $r$-junta over $\mathcal{D}$, and
        \item No: $f$ is $\frac{1}{3}$-far from every $\Omega(r\log n)$-junta over $\mathcal{D}$.
    \end{itemize}
\end{theorem}
We reduce from the {\sc SetCover} problem.
\begin{definition}[The \textsc{SetCover} problem]
    A {\sc SetCover} instance over a universe $[m]$ is a collection of subsets $\mathcal{S} = \{ S_1,\ldots,S_n\}$ where $S_i\sse [m]$. The {\sc SetCover} problem is to compute a minimal size subcollection $\{S_{i_1},\ldots, S_{i_r}\}$ which \textit{covers} the universe: $[m]=S_{i_1}\cup\cdots\cup S_{i_r}$.
\end{definition}
{\sc SetCover} is known to be hard to approximate.
\begin{theorem}[Hardness of approximating {\sc SetCover} \cite{RS97}]
    Given a {\sc SetCover} instance $\mathcal{S}$ and a parameter $r$, it is \textnormal{NP}-hard to distinguish between
    \begin{itemize}
        \item Yes: $\mathcal{S}$ has a size-$r$ set cover, and
        \item No: $\mathcal{S}$ requires set covers of size $\Omega(r\log n)$. 
    \end{itemize}
\end{theorem}

\begin{proof}[Proof of \Cref{thm:np hardness of tolerant junta testing}]
    Suppose we have an algorithm $\mathcal{T}_{\mathrm{weak}}$ for testing juntas that can distinguish between the Yes and No cases in the theorem statement. In particular, there is a $(0,\frac{1}{3},r,\Omega(r\log n))$-tester for juntas. \Cref{thm:boosting 0 error juntas} implies that there is a $(0,\smalle,r,\Omega(r\log n))$-tester, $\mathcal{T}_{\mathrm{strong}}$, for juntas as long as $\smalle$ satisfies
    \begin{equation}
        \label{eq:smalle requirement}
        \frac{1}{3}\le \frac{1-(1-2\smalle)^{(1-\lambda)k/2}}{2}\tag{$\circledast$}.
    \end{equation}
    In the reduction, we will choose $\smalle$ appropriately and use this boosted tester to solve {\sc SetCover}.
    \paragraph{The reduction.}{The reduction from {\sc SetCover} to junta testing is standard \cite{HJLT96,ABFKP08}. We will restate it here for convenience.  Let $\mathcal{S} = \{ S_1,\ldots,S_n\}$ be a {\sc SetCover} instance over the universe $[m]$ and define $u^{(1)},\ldots,u^{(m)} \in \bits^n$ where 
    \[ (u^{(j)})_i = 
    \begin{cases}
        1 & \text{if $j \in S_i$} \\
        -1 & \text{otherwise.} 
    \end{cases}
    \]
    Let $\mathcal{D}$ be the uniform distribution over $\{ u^{(1)},\ldots,u^{(m)}, (-1)^n\}$ and let $f:\bits^n\to\bits$ be the function which is the disjunction of its inputs: $f\coloneqq x_1\lor \cdots \lor x_n$ (where $1$ is interpreted as true and $-1$ as false).

    We choose $k=\Theta(m)$ so that \ref{eq:smalle requirement} holds with $\Omega(\frac{1}{m})<\smalle<\frac{1}{m+1}$. We then run the boosted tester $\mathcal{T}_{\mathrm{strong}}$ on the function $f$ and distribution $\mathcal{D}$, to test if $f$ is $0$-close to an $r$-junta or $\smalle$-far from being a $\Omega(r\log n)$-junta (where the parameters $r$ and $\Omega(r\log n)$ correspond to the {\sc SetCover} parameters). Our algorithm for {\sc SetCover} outputs Yes if and only if the tester accepts $f$ as being $0$-close to an $r$-junta.
    }
    \paragraph{Runtime.}{
        If the tester $\mathcal{T}_{\mathrm{weak}}$ runs in polynomial time, then since $k=\Theta(m)$ and $\smalle=\Theta(\frac{1}{m})$, the tester $\mathcal{T}_{\mathrm{strong}}$ runs in polynomial time. Queries to the target function $f$ and random samples from $\mcD$ can also be simulated in randomized polynomial time.
    }
    \paragraph{Correctness.}{
        For correctness, we need to show:
        \begin{itemize}[leftmargin=40pt,align=left]
            \item[\underline{Yes case}:] if $\mathcal{S}$ has a size-$r$ set cover, then $f$ is $0$-close to an $r$-junta over $\mcD$, and
            \item[\underline{No case}:] if $\mathcal{S}$ requires set covers of size $\Omega(r\log n)$, then $f$ is $\smalle$-far from being a $\Omega(k\log n)$-junta over $\mcD$.
        \end{itemize}
        \subparagraph{Yes case.}{
            Let $S_{i_1},\ldots, S_{i_r}$ be a size-$r$ set cover. Consider the function $\Tilde{f}=x_{i_1}\lor\cdots\lor x_{i_r}$. Since these indices form a set cover of $\mathcal{S}$, $\Tilde{f}(u^{(i)})=1$ for all $i\in [m]$ and $\Tilde{f}((-1)^n)=-1$. This shows $\error_{\mcD}(f,\Tilde{f})=0$. It follows that $f$ is $0$-close to an $r$-junta over $\mcD$ since $\Tilde{f}$ is an $r$-junta.
        }
        \subparagraph{No case.}{
            Suppose $\Tilde{f}$ is an $r'$-junta satisfying $\error_{\mathcal{D}}(f,\Tilde{f})< \frac{1}{m+1}$. The relevant variables of $\Tilde{f}$ must correspond to a set cover of $\mathcal{S}$: if some element $i\in [m]$ is \textit{not} covered, then $\Tilde{f}(u^{(i)})=\Tilde{f}((-1)^n)$ and $\error_{\mathcal{D}}(f,\Tilde{f})\ge \frac{1}{m+1}$. This shows if $\mathcal{S}$ requires set covers of size $\Omega(r\log n)$ then $f$ is $\frac{1}{m+1}$-far from every $\Omega(r\log n)$-junta. In particular, since $\smalle<\frac{1}{m+1}$, every $\Omega(r\log n)$-junta is $\smalle$-far from $f$.
        }
    }
\end{proof}

%% file: Counterexamples.tex
\section{Counterexamples to natural composition theorems}

\label{app:counterexamples} 

\subsection{Counterexample to Conjecture 1}
\begin{lemma}
    \label{lem:counter-1}
    For any odd $k$ and $n \geq k$ let $R = (n-1)k$ and $\mcD$ be the uniform distribution over $\bits^n$. There are symmetric functions $g:\bits^k \to \bits$ and $f:\bits^n \to \bits$ for which the following holds.
    \begin{enumerate}
        \item There is an $R$-junta $h$ achieving,
        \begin{equation*}
            \error_{\mcD^k}(g\circ f, h) \leq O(1/\sqrt{k}).
        \end{equation*}
        \item The natural strategy of dividing the budget equally achieves,
        \begin{equation*}
            \error_{\mcD^k}(g\circ f, g \circ \tilde{f}_{R/k}) = 1/2. 
        \end{equation*}
    \end{enumerate}
\end{lemma}
We set $g = \maj_k$ to be the majority function on $k$ bits,
\begin{equation*}
    g(y_1, \ldots, y_k) = \begin{cases}
        1 &\text{if }\sum_{i \in [k]} y_i \geq 0 \\
        -1 &\text{otherwise}
    \end{cases}.
\end{equation*}
and $f = \XOR_n$ to be the parity function,
\begin{equation*}
    f(x_1, \ldots, x_n) = \prod_{i \in [n]} x_i.
\end{equation*}
The following fact will be useful in giving a strategy that achieves low error.
\begin{fact}
    \label{fact:anti-concentrate}
    Let $\by_1, \ldots, \by_{k-1}$ each be uniform and independent samples from $\bits$. Then, for any choice of $c$,
    \begin{equation*}
        \Pr\bracket*{\sum_{i \in [k-1]} \by_i = c} \leq O\paren*{1/\sqrt{k}}.
    \end{equation*}
\end{fact}
We now give the junta achieving low error.
\begin{proposition}[\Cref{lem:counter-1}, there is a low error junta]
    Let $h = \maj_{k-1} \circ \XOR_n$. Then,
    \begin{enumerate}
        \item $h$ is an $((k-1)n \leq R)$-junta.
        \item $h$ achieves,
         \begin{equation*}
            \error_{\mcD^k}(g\circ f, h) \leq O(1/\sqrt{k}).
        \end{equation*}
    \end{enumerate}
\end{proposition}
\begin{proof}
    Clearly $h$ depends on only the first $(k-1)n$ bits of its inputs, so it is an $R$-junta as long as $(k-1)n \leq (n-1)k$, which is guaranteed by the assumption $n\geq k$ in \Cref{lem:counter-1}. We compute $h$'s error,
    \begin{equation*}
         \error_{\mcD^k}(g\circ f, h) = \Pr_{\by \sim \bits^n}[\maj_k(\by) \neq \maj_{k-1}(\by)].
    \end{equation*}
    In order for $\maj_k(\by) \neq \maj_{k-1}(\by)$, it must be the case that the $\sum_{i \in [k-1]} \by_i$ is $-1$ or $0$. The desired result follows from \Cref{fact:anti-concentrate}.
\end{proof}

We'll next show the natural strategy achieves advantage $0$, equivalent to error $1/2$.
\begin{proposition}
    \label{prop:xor-0-adv}
    Let $f = \XOR_n$ and $\mcD$ be the uniform distribution over $\bits^n$. Then,
    \begin{equation*}
        \adv_{\mcD}(f, \tilde{f}_{n-1}) = 0.
    \end{equation*}
\end{proposition}
\begin{proof}
    By \Cref{prop:g-avg}, it is sufficient to show that for any set $|S| = n-1$ and any $x \in \bits^n$,
    \begin{equation*}
        \Ex_{\by \sim \mcD}[f(\by) \mid \by_S = x_S] = 0.
    \end{equation*}
    For any fixed $x$, there are two $y \in \bits^n$ satisfying $y_S = x_S$: The first choice if $y = x$, and the second choice is $x$ with a single bit flipped (the one bit not in $S$). One of these two choices will have a parity of $+1$ and one will have a parity of $-1$, so the average parity is $0$, as desired.
\end{proof}

\begin{proposition}
    \label{prop:maj-NS-0}
    For any odd $k$, $\mu = 0$, and $\vrho = [0,\ldots, 0]$,
    \begin{equation*}
        \stab_{\mu, \vrho}(\maj_k) = 0.
    \end{equation*}
\end{proposition}
\begin{proof}
    For odd $k$, $\maj_k$ is an odd function, $\Ex_{\bx \sim \bits^n}[\maj_k(\bx)]$. Then,
    \begin{align*}
        \stab_{\mu, \vrho}(\maj_k) &= \Ex_{\bx_1 \sim \bits^k, \bx_2 \sim \bits^k}[\maj_k(\bx_1)\maj_k(\bx_2)] \\
        &= \Ex_{\bx_1 \sim \bits^k}[\maj_k(\bx_1)]\Ex_{\bx_2 \sim \bits^k}[\maj_k(\bx_2)] \tag{$\bx_1, \bx_2$ independent} \\
        &= 0 \cdot 0 =0.\tag{$\maj_k$ is odd}
    \end{align*}
\end{proof}

The following completes the proof of \Cref{lem:counter-1}.
\begin{corollary}
    In the setting of \Cref{lem:counter-1},
    \begin{equation*}
        \adv_{\mcD^k}(g\circ f, g \circ \tilde{f}_{R/k}) = 0. 
    \end{equation*}
\end{corollary}
\begin{proof}
    This follows from \Cref{thm:strong-formal} and \Cref{prop:xor-0-adv,prop:maj-NS-0}.
\end{proof}

\subsection{Counterexample to Conjecture 2}
\begin{lemma}
    \label{lem:counter-2}
    For any $n \geq 10$, $k \in \N$, and $R \leq n/2$, let $\mcD$ be uniform over $\bits^n$. There are $g: \bits^k$ and $f:\bits^n \to \bits$ for which, for all partitions $r_1 + \cdots +r_k = R$,
    \begin{equation*}
        \error_{\mcD^k}(g\circ f, g(\tilde{f}_{r_1}, \ldots, \tilde{f}_{r_k})) \geq 1 - 2^{-\Omega(k)}.
    \end{equation*}
\end{lemma}
\Cref{lem:counter-2} is particularly surprising in light of the fact that either the constant $-1$ or constant $1$ functions, both of which are $0$-juntas, will achieve error $\leq 1/2$ with respect to $g \circ f$. We begin with a probabilistic construction of $f$ achieving the following.
\begin{proposition}
    \label{prop:probabilistic-f}
    For any $n \geq 10$, there is an $f: \bits^n \to \bits$ for which $\Ex_{\bx \sim \bits^n}[f(\bx)] \leq 0.5$ but, for all $|S| \leq n/2$ and $x \in \bits^n$,
    \begin{equation*}
        \Ex_{\bx \sim \bits^n}[f(\bx) \mid \bx = x] > 0.
    \end{equation*}
\end{proposition}
\begin{proof}
    Consider a random function $\boldf$ where, for each $x \in \bits^n$, $\boldf(x) \sim \pi_{0.25}$. We'll show that $\boldf$ meets the desired criteria with a strictly positive probability, proving the existence of at least one such $f$.

    Let $\mu(\boldf) \coloneqq \Ex_{\bx \sim \bits^n}[\boldf(\bx)]$. Then $\mu(\boldf)$ is the average of $2^n$ independent samples of $\pi_{0.25}$. Applying Hoeffding's inequality,
    \begin{equation*}
        \Pr[\mu(\boldf) > 0.5] \leq \exp(-2 \cdot (0.25)^2 \cdot 2^n) = \exp(-2^n/2).
    \end{equation*}
    Similarly, for any $|S| \leq n/2$ and $x \in \bits^n$, let $\mu(\boldf, S, x) \coloneqq  \Ex_{\bx \sim \bits^n}[\boldf(\bx) \mid \bx = x]$. $\mu(\boldf,S,x)$ the average of at least $2^{n/2}$ independent samples of $\pi_{0.25}$. Once again, by Hoeffding's inequality,
     \begin{equation*}
        \Pr[\mu(\boldf,S,x) \leq 0] \leq \exp(-2 \cdot (0.25)^2 \cdot 2^{n/2}) = \exp(-2^{n/2}/2).
    \end{equation*}
    Union bounding over all $2^n$ choices of $S$ and $2^n$ choices for $x$, we have that $\boldf$ meets the desired criteria with probability at least
    \begin{equation*}
        1 - \exp(-2^n/2) - 2^{2n}\exp(-2^{n/2}/2).
    \end{equation*}
    When $n \geq 10$, the above probability is strictly positive, so such an $f$ must exist.
\end{proof}

\begin{proof}[Proof of \Cref{lem:counter-2}]
    Let $f$ be a function with the properties of \Cref{prop:probabilistic-f}, and $g = \textsc{And}_k$ return $+1$ if and only if all $k$ of its inputs are $+1$. By \Cref{prop:g-avg}, for any $r \leq n/2$, $\tilde{f}_r$ is the constant $+1$ function. Therefore, for any $r_1 + \cdots + r_k = R$, $g(\tilde{f}_{r_1}, \ldots, \tilde{f}_{r_k})$ is the constant $+1$ function. However,
    \begin{equation*}
        \Pr_{\bx \sim \mcD^k}[(g \circ f)(\bx) = +1] = (3/4)^k. \qedhere
    \end{equation*}
\end{proof}

\subsection{Counterexample to Conjecture 3}
\begin{lemma}
    \label{lem:counter-3}
    There is $g:\bits^k \to \bits$, $f:\bits^n \to \bits$, distribution $\mcD$ over $\bits^n$, and budget $R$ for which no $R$-junta of composed form achieves optimal error among all $R$-Juntas for $g\circ f$ with respect to $\mcD^k$.
\end{lemma}



\begin{proof}
We'll set $k = 2$, $g = \textsc{And}_2$. Let $p:\bits^2 \to [0,1]$ be defined as
\begin{equation*}
    p(x) \coloneqq \begin{cases}
        1 &\text{if }x_1 = x_2 = 1,\\
        3/4 &\text{if } x_1 \neq x_2,\\
        3/5 &\text{if } x_1 = x_2 = -1.
    \end{cases}
\end{equation*}
We begin by describing a probabilistic construction: Given the input $x$, the value of $\boldf(x)$ will still be a random variable. In particular, we set $n =2$, and $\boldf(x)$ is set to $+1$ with probability $p(x)$ and $-1$ otherwise. This probabilistic construction will later be derandomized. We allow a junta budget of $R = 4$.

Next, we construct an optimal approximator for $g \circ \boldf$. Given an input $x^{(1)}, x^{(2)}$, let $\by_1 = \boldf(x^{(1)})$ and $\by_2 = \boldf(x^{(2)})$. For succinctness, we'll use $p_i$ to refer to the $\Pr[\by_i = 1]$. Then, since $g = \textsc{And}_2$, the optimal approximator will return $1$ iff $p_1p_2 \geq 1/2$. For our particular $\boldf$ the only choices for $p_i$ are $3/5,3/4,1$. As a result,
\begin{equation*}
    h^{(\mathrm{opt})}(p_1,p_2) = \begin{cases}
        1 &\text{if }p_1 = 1\text{ or }p_2 = 1,\\
        1 &\text{if }p_1 = p_2 = 3/4,\\
        0 &\text{otherwise}.
    \end{cases}
\end{equation*}
However, no composed form can achieve the above optimal approximator. Recall that composed form approximators are of the form $h(q_1, q_2)$, where each $q_i$ has range $\bits$. The fact that the size of this range is $2$, but there are three possible choices $(3/5, 3/4, 1)$ for $p_i$, is the crux of the issue.

In more detail, of the three choices $(3/5,3/4,1)$ for $p_i$, $q_1$ must classify at least two of them the same way. This gives three cases.
\begin{enumerate}
    \item If $q_1$ classifies $3/4$ and $1$ the same way, $h(q_1, q_2)$ cannot distinguish between $p_1 = 3/4, p_2 = 3/5$ and $p_1 = 1, p_2 = 3/5$, and so cannot be optimal.
    \item If $q_1$ classifies $3/5$ and $3/4$ the same way, $h(q_1, q_2)$ cannot distinguish between $p_1 = 3/4, p_2 = 3/4$ and $p_1 = 3/5, p_2 = 3/4$, and so cannot be optimal.
    \item If $q_1$ classifies $3/5$ and $1$ the same way, $h(q_1, q_2)$ cannot distinguish between $p_1 = 3/5, p_2 = 3/4$ and $p_1 = 1, p_2 = 3/4$, and so cannot be optimal.
\end{enumerate}
In all three cases composed form cannot achieve optimal error. It will always be off by some constant. 

To derandomize this construction, we set $n \gg 2$ sufficiently large. For each $x \in \bits^n$, we sample the value $f(x)$ to be $+1$ with probability $p(x_1,x_2)$ and $-1$ otherwise. Note that after randomly selecting the value of $f$ on each input $x \in \bits^n$, $f$ is now a deterministic function. Following the same arguments as in \Cref{prop:probabilistic-f}, with high probability over the random choices in defining $f$, the error of the optimal $4$-junta and of the optimal composed form $4$-junta for $g\circ f$ are within $\pm \eps(n)$ of what they are for $g \circ \boldf$, where $\eps(n)$ goes to $0$ as $n \to \infty$. Therefore, for sufficiently large $n$, there exists an $f$ meeting the desired criteria.
\end{proof}